%% file: mpq0.tex
\documentclass[reqno,10pt]{amsart}

\usepackage{amssymb,amsthm,amsmath,amsfonts}
\usepackage{amsaddr}

\input{mpq-macro3}

\numberwithin{equation}{section}

\begin{document}
\title[Stability of Benney-Luke solitary waves]
{Asymptotic stability of solitary waves in the Benney-Luke model of
water waves }
\author{Tetsu Mizumachi}
\address{
Faculty of Mathematics, Kyushu University, Fukuoka 812-8581, Japan}

\author{Robert L. Pego}
\address{Department of Mathematical Sciences, Carnegie Mellon University,
Pittsburgh, PA, USA}

\author{ Jos\'e Ra\'ul Quintero}
\address{Departamento  Matem\'aticas, Universidad del Valle, Colombia}

\begin{abstract}
We study asymptotic stability of solitary wave solutions
in the one-dimensional Benney-Luke equation, a formally
valid approximation for describing two-way water wave propagation.
For this equation, as for the full water wave problem,
the classic variational method for proving
orbital stability of solitary waves fails dramatically due to 
the fact that the second variation of the energy-momentum functional
is infinitely indefinite.
We establish nonlinear stability in
energy norm under the spectral stability hypothesis
that the linearization admits no non-zero eigenvalues of
non-negative real part. 
We then verify this hypothesis for waves of small energy.
\end{abstract}
\maketitle





\vfil
\pagebreak

\setcounter{tocdepth}{1}
\tableofcontents

\input{mpq-intro}

\part{Spectral and linear stability}

\input{mpq-part1}

\part{Nonlinear stability}

\input{mpq-part2}

\appendix
\input{mpq-ham}

\input{mpq-appx1}

\input{mpq-appx2}

\bigskip\noindent{\bf Acknowledgments:}
This material is based upon work supported by the National Science
Foundation under grant nos.\ DMS 0604420 and DMS 0905723,
and partially supported by the Center for Nonlinear Analysis (CNA) 
under the National Science Foundation Grant 
no. 0635983 and PIRE Grant no. OISE-0967140.
TM is supported by Grant-in-Aid for Scientific Research no. 21540220.
The work of JRQ is supported by the Universidad del Valle (Cali,
Colombia).

\medskip

\bibliographystyle{siam}
\bibliography{mpq}

\end{document}

%% file: mpq-macro3.tex
\newcommand{\nwc}{\newcommand}

\newtheorem{theorem}{Theorem}[section]
\newtheorem{proposition}[theorem]{Proposition}
\newtheorem{lemma}[theorem]{Lemma}
\newtheorem{corollary}[theorem]{Corollary}

\newtheorem{claim}{Claim}[section]
\theoremstyle{definition}

\nwc{\todo}[1]{\smallskip \begin{center}{TODO: #1}\end{center} \smallskip}
\nwc{\toprove}{\qquad\mbox{YET TO PROVE}}
\nwc{\CHECK}[1]{[[CHECK #1]]}

\nwc{\hide}[1]{}

\nwc{\bit}{\begin{itemize}}
\nwc{\eit}{\end{itemize}}

\nwc{\qref}[1]{(\ref{#1})}
\nwc{\ip}[1]{\langle{#1}\rangle}
\nwc{\ipbig}[1]{\left\langle{#1}\right\rangle}
\nwc{\pmat}[1]{ \begin{pmatrix} #1 \end{pmatrix} }
\nwc{\pfrac}[2]{\left(\frac{#1}{#2}\right)}

\nwc{\bm}{\boldmath}
\nwc{\ti}{\tilde}

\nwc{\half}{{\textstyle \frac12}}
\nwc{\sfrac}[2]{{\textstyle \frac{#1}{#2}}}
\nwc{\intR}{\int_\R}

\nwc{\inv}{^{-1}}
\nwc{\R}{{\mathbb{R}}}
\nwc{\C}{{\mathbb{C}}}
\nwc{\D}{\partial}
\nwc{\dt}{\partial_t}
\nwc{\dx}{\partial_x}
\nwc{\dc}{\partial_c}
\nwc{\di}{\partial\inv}
\nwc{\dhx}{\partial_{\hat x}}

\nwc{\one}{{\mathbbm{1}}}
\nwc{\Bo}{\mathop{\rm Bo}\nolimits}
\nwc{\re}{\mathop{\rm Re}\nolimits}
\nwc{\im}{\mathop{\rm Im}\nolimits}

\newcommand{\cal}{\mathcal}
\nwc{\Ac}{{\cal A}}
\nwc{\SA}{{\cal A}}
\nwc{\BB}{{\cal B}}
\nwc{\CC}{{\cal C}}
\nwc{\EE}{{\cal E}}
\nwc{\FF}{{\cal F}}
\nwc{\GG}{{\mathcal G}}
\nwc{\HH}{{\cal H}}
\nwc{\II}{{\cal I}}
\nwc{\JJ}{{\cal J}}
\nwc{\KK}{{\cal K}}
\nwc{\LL}{{\cal L}}
\nwc{\MM}{{\cal M}}
\nwc{\NN}{{\cal N}}
\nwc{\PP}{{\cal P}}
\nwc{\QQ}{{\cal Q}}
\nwc{\RR}{{\cal R}}
\nwc{\SC}{{\cal S}}
\nwc{\TT}{{\cal T}}
\nwc{\UU}{{\cal U}}
\nwc{\VV}{{\cal V}}
\nwc{\WW}{{\cal W}}
\nwc{\XX}{{\cal X}}
\nwc{\YY}{{\cal Y}}
\nwc{\ZZ}{{\cal Z}}

\nwc{\ep}{{\epsilon}}
\nwc{\aw}{\alpha}
\nwc{\lam}{\lambda}
\nwc{\bphi}{\varphi}

\nwc{\bondi}{\hat\tau} 
\nwc{\sw}{\Theta}

\newcommand\hQ{\hat{\mathcal{Q}}}
\newcommand\hS{\hat{\mathcal{S}}}
\nwc{\hap}{{\hat\alpha}}
\nwc{\hxi}{\hat\xi}
\nwc{\hk}{\hat k}
\nwc{\hbet}{{\hat\alpha_\ep}}
\nwc{\hb}{{\hat\beta}}
\nwc{\hep}{\hat\ep}

\nwc{\mo}{m_0}
\nwc{\mep}{m_\ep}
\nwc{\thep}{\theta_\ep}
\nwc{\thez}{\theta_0}
\nwc{\isc}{{\mathcal I}_\ep}
\nwc{\Omst}{\Omega_*}

\nwc{\Bi}{B^{-1}}
\nwc{\thc}{\Theta_c}
\newcommand\ap{\alpha}
\newcommand\apc{\alpha_c}

\newcommand\V{{\mathcal{V}}}

\nwc{\calM}{{\mathcal{M}}}
\nwc{\calN}{{\mathcal{N}}}
\nwc{\calR}{{\mathcal{R}}}
\nwc{\calS}{{\mathcal{S}}}

\newcommand\CH{{\mathcal{H}}}
\newcommand\CN{{\mathcal{N}}}
\newcommand\CJ{{\mathcal{J}}}
\newcommand\CT{{\mathcal{T}}}

\newcommand\Q{{\mathcal{Q}}}

\newcommand\N{\mathbb{N}}

\newcommand\z{\zeta}

\newcommand\F{{\mathcal{F}}}
\newcommand\E{{\mathcal{E}}}

\nwc{\BBo}{L+c\D} 

\nwc{\slap}{_{L^2_\ap}}
\nwc{\lalplalp}{L^2_\alpha\times L^2_\alpha}
\nwc{\Lc}{{\mathcal{L}}_c}
\nwc{\Lct}{{\mathcal{L}}_{c(t)}}
\nwc{\Lco}{{\mathcal{L}}_{c_0}}

\nwc{\ww}{\WW} 
\nwc{\wep}{\WW_\ep}

\theoremstyle{remark}
\newtheorem{remark}{Remark}[section]
\newcommand{\la}{\langle}
\newcommand{\ra}{\rangle}
\newcommand{\pd}{\partial}

\newcommand{\bM}{\mathbb{M}}
\nwc{\bmc}{\bM_c}
\nwc{\bmx}{\bM_x}
\nwc{\bmv}{\bM_v}
\nwc{\bmone}{\bM_1}
\nwc{\bmtwo}{\bM_2}
\nwc{\bmtot}{\bM_{\rm tot}}
\nwc{\kone}{k_1}
\nwc{\ktwo}{k_2}

\DeclareMathOperator{\sech}{sech}

\DeclareMathOperator{\spann}{span}
\renewcommand{\Re}{\re} 
\renewcommand{\Im}{\im} 

\nwc{\zone}{\zeta_{1,c}}
\nwc{\ztwo}{\zeta_{2,c}}
\nwc{\zzone}{\eta_{1,c}}
\nwc{\zztwo}{\eta_{2,c}}
\nwc{\zaone}{\zeta_{1,c}^*}
\nwc{\zatwo}{\zeta_{2,c}^*}

\nwc{\vonew}{\|\ti v_1(t)\|_{W_{\vap}}}
\nwc{\vonews}{\|\ti v_1(s)\|_{W_{\vap}}}
\nwc{\wnorm}[1]{\|#1\|_{W_{\vap}}}

\nwc{\gwt}{\omega}
\nwc{\cp}{{c_\star}} 
\nwc{\xp}{{x_\star}} 
\nwc{\xot}{x_0(t)}
\nwc{\vap}{\nu}
\nwc{\apo}{{\ap_1}}
\nwc{\tcha}{\tilde{\chi}}
\nwc{\cha}{\chi_\vap}
\nwc{\rp}{\rho}

%% file: mpq-intro.tex
\section{{Introduction}}
In this paper we study the stability  of solitary
waves for the nonlinear dispersive wave equation
\begin{equation}
\dt^2\bphi
-\dx^2\bphi
+a\dx^4\bphi
-b\dx^2\dt^2\bphi
+(\dt\bphi)(\dx^2\bphi)+2(\dx\bphi)(\dx\dt\bphi)=0.
\label{e.bl}
\end{equation}
This equation is a formally valid approximation for describing 
small-amplitude, long water waves in water of finite depth.
It is the one-dimensional version of an equation originally derived 
by Benney and Luke \cite{BL64} as an isotropic model 
for three-dimensional water waves.  (Also see \cite{PQ99}.) 
The parameters $a, \ b >0$ are  such that 
$a-b = \bondi -\frac13$, where $\bondi$ is the inverse Bond number.
We will assume $0<a<b$ throughout this paper,
which corresponds to small or zero surface tension ($\bondi<\frac13$), 
and will study solitary waves that travel with any speed 
$c$ satisfying $c^2>1$. 
These waves are captured in the small-amplitude, long-wave regime 
of formal validity for speed with $c^2$ near 1. 
We remark that \qref{e.bl} is an approximation formally valid 
for describing {\em two-way} water wave propagation, in contrast to
one-way equations such as the KdV, BBM, or KP equations.  

There are a considerable number of previous works on 
model water wave equations and stability of solitary 
waves---Especially see \cite{BCS02,BCL05,CCD10}
and references therein regarding a large family of models 
of Boussinesq type.
Our interest in the Benney-Luke equation \qref{e.bl} is
motivated by a number of features that this equation shares with the 
full water wave equations with zero surface tension, for which the 
problem of solitary wave stability remains open.  

In particular, many works on model equations employ the variational
method originated by Benjamin \cite{Ben72} and Bona \cite{Bo75} for
the KdV equation and developed into a powerful abstract tool by 
Grillakis, Shatah and Strauss \cite{GSS1,GSS2}.
In particular, for \qref{e.bl} in the complementary case corresponding to strong
surface tension ($a>b$ and $0\le c^2<{a/b}$), Quintero
has used this method to establish orbital stability of solitary waves,
 in both one and two space dimensions \cite{Q03,Q05}.
Techniques related to the GSS variational method have been developed 
to obtain many further results as well.
For the full water wave equations with strong surface tension
($\bondi>\frac13$), Mielke has obtained a
conditional orbital stability result for small solitary waves of
depression~\cite{Mi02}.  
For generalized KdV equations, a variety of results concerning asymptotic
stability with small perturbations of finite energy, 
even for multiple pulses,
have been obtained by Martel and Merle and others, 
using various methods involving 
variational, virial, and nonlinear Liouville properties,
e.g., see \cite{MM00,MM01,MMT02,MM05,MM11,Miz04,El05b}. 

For the Benney-Luke equation \qref{e.bl} in the present case---as for the 
full water wave equations with zero surface tension
\cite{BS89}---the variational approach fails dramatically, 
due to the infinite indefiniteness of the energy-momentum functional
whose critical points are the solitary-wave profiles. 
(See Appendix \ref{ap-ham} below for details.)
This also happens for the line soliton solution of the KP-II equation
with periodically transverse perturbations, which was recently 
studied by Mizumachi and Tzvetkov \cite{MT11} using methods based on
Miura transformations.  But the Benney-Luke and water wave equations 
are nonintegrable, and such transformation techniques appear to be 
unavailable.

A salient feature of solitary water waves that is shared 
by solitary waves of \qref{e.bl} in the present case with
$0<a<b$ and $c^2>1$ is that 
they travel at a speed greater than the maximum group velocity of linear waves. 
This motivates us to study the scattering of localized perturbations
by using norms with spatial weights that decay to zero in the 
direction behind solitary waves.
This approach has been used successfully to obtain
asymptotic stability results for the 
KdV equation \cite{PW94}, the BBM equation \cite{MW96}, and 
Fermi-Pasta-Ulam lattice equations \cite{FP2,FP3,FP4},
for pertubations small both in energy and in weighted norm with
exponential weights $e^{ax}$.
And, in a recent analysis for the Toda lattice \cite{Miz09}, 
Mizumachi established asymptotic
stability of solitary waves for arbitrary perturbations of small energy, 
by combining stability estimates with exponential weights, 
as used by Friesecke and Pego,
with dispersive propagation estimates (virial estimates) 
related to techniques of Martel and Merle.

In the present paper, we will build on Mizumachi's approach to establish
asymptotic stability results for Benney-Luke solitary waves of small amplitude.
The analysis comes in two main parts, corresponding to linear and
nonlinear analysis.

In the first part, we show that spectral stability---the absence 
of nonzero eigenvalues with non-negative real part---implies 
linear stability with exponential decay rate in exponentially weighted norm, 
for perturbations orthogonal to the adjoint neutral-mode space
generated by variations of phase and wave amplitude.
Also, we prove that small solitary waves are spectrally stable.
This is done by a suitable comparison of a reduced resolvent
operator with the corresponding one for KdV solitons.
The KdV limit is used to control the reduced resolvent on long length 
scales at long times, and this is combined with additional estimates
to obtain control on all length scales and all time scales. 
This happens in a manner similar to the spectral stability analysis 
of small solitary waves in FPU lattices \cite{FP4} and water waves \cite{PS2}.

In the second part, we prove that nonlinear stability follows 
from spectral stability.  Perturbed solitary waves are studied in terms of 
i) time-modulated speed and shift parameters, 
ii) a solution freely propagated from the intial wave perturbation,
and iii) the exponentially localized interaction of the two.
Key to this analysis is a linear decay estimate
based on a discrete-time {\em recentering} technique
reminiscent of the {renormalization} method developed by Promislow
\cite{Pr02} for studying pulse dynamics in parametrically driven
nonlinear Schr\"odinger equations. 
For the present problem, discrete-time recentering is used to avoid a
problem of loss of derivatives in linear stability estimates that may
occur in frames translating with time-varying speed.

\section{Statement of main results} \label{s.main}

\subsection{Equations of motion}
In terms of the notation
\[
 q=\dx\bphi, \quad r=\dt\bphi, 
 \qquad A=I-a\dx^2, \quad B=I-b\dx^2, 
\]
the Benney-Luke equation takes the form of the system
\begin{equation}\label{e.qr}
\dt q -\dx r=0,\qquad B \dt r -A\dx q + r\dx q+2q \dx r = 0.
\end{equation}
We write this system in the abstract form
\begin{equation}\label{eq:bousys}
\dt u = Lu + f(u) ,
\end{equation}
with 
\begin{equation} \label{d.Lf}
u=\begin{pmatrix} q \\ r\end{pmatrix},
\quad
 L=\begin{pmatrix} 0 & \pd_x \\ B^{-1}A\pd_x & 0 \end{pmatrix},
\quad 
 f(u)=\begin{pmatrix}0 \\ -B^{-1}(r\pd_xq+2q\pd_xr)\end{pmatrix}.
\end{equation}
The energy
\begin{equation}
E(u)=
\frac12 \int_{\R}
\left(q^2+r^2+a(\D_xq)^2+b(\D_xr)^2\right)dx = 
\frac12 \int_\R (qAq+rBr)\,dx
\end{equation}
is formally conserved in time for solutions to \eqref{eq:bousys}:
\begin{equation}
  \label{eq:energy}
  E(u(t))=E(u_0).
\end{equation}
By standard arguments, the Cauchy problem 
for \qref{e.qr} is globally well-posed for initial data in the Sobolev
space $H^s(\R)$ for any $s\ge1$, and \qref{eq:energy} holds for all $t$. 

\subsection{Solitary waves}
The Benney-Luke system \qref{e.qr} admits a two-parameter family
of solitary waves 
\[
(q,r)=(q_c(x-ct-x_0),r_c(x-ct-x_0)), 
\qquad c^2>1, \quad x_0\in\R,  
\]
whose profiles must satisfy
\begin{equation}\label{e.sol0}
-c q_c'-r_c' = 0, \qquad
-cB r_c' - A q_c' + r_cq_c'+2q_cr_c' = 0,
\end{equation} whence
\begin{equation}\label{e.sol}
r_c=-c q_c, \qquad (bc^2-a)q_c''-(c^2-1)q_c+\frac{3c}2 q_c^2 = 0.
\end{equation}
Explicitly,
\begin{equation}\label{e.sol2}
q_c(x)= \frac{c^2-1}{c}\sech^2 \left(\frac12 \apc x\right),
\qquad \apc =\sqrt{\frac{c^2-1}{bc^2-a}}.
\end{equation}

\subsection{Spectral and linear stability}
We linearize the Benney-Luke system \qref{e.qr} 
about a solitary wave $(q,r)=(q_c,r_c)$ with $c>1$,
after changing to a coordinate frame moving with speed $c$.
The resulting linearized system for the perturbation
$z=(q_1,r_1)$ takes the following form, with $\D=\D_x$:
\begin{equation}\label{e.lin1}
\D_t z = \Lc z, \qquad
\Lc = 
\begin{pmatrix}
c\D & \D \cr -B\inv (-A\D + r_c\D+2r_c') & 
c\D-B\inv(2q_c\D+q_c') 
\end{pmatrix}.
\end{equation}
We study equation \qref{e.lin1} in exponentially weighted 
function spaces, writing 
\[
L^p_\alpha=\{g\mid e^{\ap x}g\in L^p(\R)\},
\qquad H^s_\alpha=\{g\mid e^{\ap x}g\in H^s(\R)\},
\]
with norms 
\begin{equation*}
 \|g\|_{L^p_\alpha}=\left(\int_\R |e^{\alpha x}g(x)|^p\,dx\right)^{1/p}\,,
\quad  
\|g\|_{H^s_\alpha}=\left(\int_\R 
(1+|k|^2)^s |\hat g(k+i\ap)|^2\,\frac{dk}{2\pi}
\right)^{1/2}.
\end{equation*}
In these definitions, $g$ may be scalar or vector valued according
to context.
We normalize the Fourier transform according to the definition
\[
\hat g(k) = \int_\R g(x)e^{-ikx}\,dx.
\]
Note that for any $s\ge0$, 
the solitary wave profile components $q_c$, $r_c\in H^s_\ap$ 
if and only if $|\ap|<\apc$.
In terms of these spaces, the following basic facts will be established
in section \ref{s.Lc}.
\begin{lemma} \label{l.lin1}
Assume $0<a<b$ and $c>1$.  Fix $\alpha$ with $0<\ap<\apc$,
and consider the operator $\Lc$ from \qref{e.lin1} 
in the space $L^2_\ap$ with domain 
$D(\Lc)=H^1_\ap$.  Then
\begin{itemize}
\item[(i)]  $\Lc$ is a compact perturbation of $\BBo$.
\item[(ii)] $\Lc$ 
is the generator of a $C^0$ semigroup in $L^2_\ap$.
\item[(iii)] The essential spectrum of $\Lc$ is contained strictly 
in the left half-plane $\Re \lambda<0$.
\item[(iv)]  The value $\lambda=0$ is an eigenvalue of $\Lc$ with 
multiplicity 2. Specifically, 
\begin{equation} \label{d.tan1}
\z_{1,c} = \begin{pmatrix}\D_x q_c\cr \D_x r_c\end{pmatrix}, 
\qquad
\z_{2,c} = -\begin{pmatrix}\D_c q_c\cr \D_c r_c\end{pmatrix},
\end{equation}
satisfy $\Lc\z_{1,c}=0$, $\Lc\z_{2,c}=\z_{1,c}$.
\end{itemize}
\end{lemma}

We let $P_c$ denote the spectral projection onto 
$\spann\{\zone,\ztwo\}$,
the generalized eigenspace associated with the eigenvalue 0 for $\Lc$.
Our main results in Part I 
(concerning asymptotic linear stability) are as follows.
\begin{theorem}[Spectral stability implies linear stability]
\label{t.splin} 
Fix $c>1$ and $\alpha$ with $0<\ap<\apc$.
Assume that $\Lc$ has no nonzero eigenvalue $\lambda$
satisfying $\Re\lambda\ge 0$. 
Then there exist positive constants $K$ and $\beta$ such that
for all $z\in L^2_\ap$ and all $t\ge0$, 
\begin{equation}\label{i.semi1}
\|e^{\Lc t}(I-P_c) z\|\slap \le Ke^{-\beta t}\|z\|\slap.
\end{equation}
\end{theorem}

\begin{theorem}[Spectral stability for small waves]
\label{t.smlin}
Fix $\hap\in(0,(b-a)^{-1/2})$. 
Then there exists $\ep_0>0$ such that whenever $0<\ep<\ep_0$
and $c=1+\frac12\ep^2$,
then $\Lc$ has no nonzero eigenvalue $\lambda$ 
satisfying $\Re\lambda\ge0$, 
in the space $L^2_{\ep\hap}$. 
\end{theorem}

\subsection{Nonlinear stability}
Our basic nonlinear stability result establishes asymptotic orbital 
stability for the family of solitary waves, with respect to arbitrary 
small-energy perturbations.
\begin{theorem} [Spectral stability implies nonlinear stability]
  \label{thm:1}
Suppose $c_0>\sigma>1$ and $0<\ap<\frac12\ap_{c_0}$, 
and assume that in $L^2_\ap$,
$\Lco$ has no nonzero eigenvalue $\lambda$ satisfying $\Re\lambda\ge0$. 
Then there exists $\delta>0$ satisfying the following:
If $u_0(x)=u_{c_0}(x-x_0)+v_0(x)$ where $x_0\in\R$ and
$\|v_0\|_{H^1}<\delta$, then there exist $\cp>1$ and
a $C^1$-function $x(t)$ such that
\begin{gather}
\label{Phase2}
|\cp-c_0|+\sup_{t>0}|{x'}(t)-c_0|=O(\|v_0\|_{H^1}),
\\  
\label{Phase1}
\lim_{t\to\infty}{x'}(t)=\cp,
\\ 
\label{OS}
\sup_{t\ge0}\|u(t,\cdot)-u_{c_0}(\cdot-x(t))\|_{H^1}^2=O(\|v_0\|_{H^1}),
\\ 
\label{AS}
\lim_{t\to\infty}\|u(t,\cdot)-u_{\cp}(\cdot-x(t))\|_{H^1(x\ge \sigma t)}=0.
\end{gather}
\end{theorem}

If the initial perturbation 
is sufficiently localized ahead of the main solitary wave,
then we obtain 
local convergence to some fixed solitary wave,
and an estimate of the local decay rate. 
\begin{theorem} [Asymptotic phase and local decay rates]
\label{thm:2}
In addition to the assumptions of Theorem \ref{thm:1}, assume that
\begin{equation}
  \label{eq:g}
\|\gwt v_0\|_{L^2(\R)} +\|\gwt\pd_xv_0\|_{L^2(\R)} <\infty\,,  
\end{equation}
where $\gwt $ is an increasing function on $\R$ such that
$\gwt(x)=1$ for $x\le 0$ and $1/\gwt(x)$ is integrable on $(0,\infty)$.
Then 
\begin{equation}
  \label{Phase3}
\xp= \lim_{t\to\infty}(x(t)-\cp t) \qquad\mbox{exists.}  
\end{equation}
Additionally, 
\begin{enumerate}
\item if $v_0\in H^1_\apo$ for some $\apo>0$ small, then there exists
$\gamma>0$ such that as $t\to\infty$,
  \begin{gather}
\label{Exp-phase-convergence}
|x(t)-\cp t-\xp|=O(e^{-\gamma t}),
\\ 
    \label{Exp-convergence}
\left\|u(t,\cdot+\cp t+\xp)-u_{\cp }\right\|_{H^1_\apo}=O(e^{-\gamma t});
  \end{gather}
\item if $\int_0^\infty x^{2\rp}(|v_0(x)|^2+|\pd_xv_0(x)|^2)\,dx<\infty$
for some $\rp>1$, then as $t\to\infty$,
\begin{equation}
  \label{Poly-convergence}
\left\|\min(1,e^{\alpha x})
(u(t,\cdot+\cp t+\xp)-u_{\cp })\right\|_{H^1}=O(t^{-\rp+1}).
\end{equation}
\end{enumerate}
\end{theorem}

A linear stability estimate of particular significance 
in the proof of nonlinear stability is stated in the following lemma,
which is used to deal with time-dependent variations of wave speed and phase. 
The proof, provided in Appendix~\ref{sec:appendix-L},
involves the recentering technique mentioned in the introduction,
in order to avoid estimating the time-dependent advection term 
as a forcing term.
\begin{lemma}
  \label{lem:lineardecay2}
Let $c_0>1$ and $0<\alpha<\alpha_{c_0}$. 
Assume that in $L^2_\ap$, $\Lc$ has no nonzero eigenvalue $\lambda$ satisfying 
$\Re\lambda\ge0$.  Then there exist positive constants
$\hat\delta$ and $\hat K$ 
with the following property.
Suppose $c(t)$ and $\eta(t)$ are continuous functions on $[0,T)$ 
$(0<T\le\infty)$ such that 
\begin{equation}\label{a.ceta}
\sup_{t\in[0,T)}\left(|c(t)-c_0|+|\eta(t)|\right)<\hat\delta.
\end{equation}
Suppose also that $F\in C([0,T);H^1_\alpha)$, and that
$w \in C([0,T);H^1_\alpha)$ is a solution of 
\begin{equation}
  \label{eq:linearw}
 \pd_tw=\mathcal{L}_{c(t)}w+\eta(t)\pd_yw+F(t),
\end{equation}
satisfying the non-secularity condition $P_{c(t)}w(t)=0$, $t\in[0,T)$.
Then
\begin{equation}
  \label{eq:w-estimate}
\|w(t)\|_{H^1_\alpha}\le \hat K\left(e^{-\beta t/3}\|w(0)\|_{H^1_\alpha}
+\int_0^t e^{-\beta(t-s)/3}\|F(s)\|_{H^1_\alpha}\,ds\right)\,,
\end{equation}
where $\beta$ is the constant given in Theorem~\ref{t.splin} for $c=c_0$.
\end{lemma}

%% file: mpq-part1.tex
To prove Theorem~\ref{t.splin}, our plan is to use the 
characterization of exponential stability provided by the 
Gearhart-Pr\"uss theorem \cite{Ge78,Pr85}.
(See \cite{CL03} for a survey of the use of this theorem, 
and also \cite{AGG86,Herbst,How,Huang}.)
By this theorem (in particular see \cite[Cor. 4]{Pr85}), 
exponential stability of a $C^0$ semigroup in a Hilbert
space is equivalent to the uniform boundedness of the resolvent on the 
right half-plane.
We apply this theorem in the space $Z_\ap:=(I-P_c)L^2_\ap(\R,\R^2)$, 
the spectral complement of the neutral-mode space.
Thus, the conclusion on linear stability in Theorem~\ref{t.splin}
is equivalent to the statement that
the restricted resolvent $(\lambda-\Lc)\inv|_{Z_\ap}$ 
is uniformly bounded on the right half-plane
 $\C_+=\{\lambda\in\C:\Re\lambda>0\}$.  

Naturally, the restricted resolvent is bounded
for $\lambda$ in a neighborhood of the discrete eigenvalue 0. 
And due to Lemma~\ref{l.lin1}, the main hypothesis of
Theorem~\ref{t.splin} ensures that the resolvent set of $\Lc$ contains
all of the closed right half-plane 
$\bar\C_+$ 
except $\lambda=0$. 
The restricted resolvent is therefore bounded on compact subsets  
of $\bar\C_+$. To complete the proof, then, it will suffice to prove 
that the resolvent is bounded on $L^2_\ap$,
uniformly outside a bounded set in $\C_+$. 
That is, there are constants $M_1$ and $M_2$ such that
\begin{equation}\label{i.unif1}
\sup_{\Re\lambda>0,|\lambda|>M_2}
\|(\lambda-\Lc)\inv \|_\ap\le M_1 .
\end{equation}
(Throughout this part we write $\|\cdot\|_\ap$ to denote the operator norm 
in $L^2_\ap$.) The proof of \qref{i.unif1} will be completed in
section~\ref{s.splin}.

The proof of Theorem~\ref{t.smlin} involves analysis 
of eigenvalues in the KdV scaling limit.
The eigenvalue problem for $\Lc$ is reduced to a 
characteristic value problem (nonlinear eigenvalue problem)
for an analytic Fredholm operator bundle $\ww(\lambda)$,
for  which the value $\lambda=0$ has ``null multiplicity'' at least 2.
In the KdV limit, this bundle converges after scaling to one naturally
associated with the KdV soliton, for which the only characteristic value is at
the origin, with null multiplicity exactly 2.  Theorem~\ref{t.smlin} will
be proved using the operator-valued version of Rouch\'e's theorem due to
Gohberg and Sigal \cite{GS} to conclude that $\ww(\lambda)$ can have no nonzero 
characteristic values $\lambda\ne0$ satisfying $\Re\lambda\ge0$.

To begin all the analysis, it is convenient to 
change variables to diagonalize $\BBo$, the leading part of the system. 
Define the Fourier multiplier operators
\begin{equation}
\calS = \sqrt{B\inv A}= \sqrt{\frac{1-a\D^2}{1-b\D^2}}, 
\qquad
\Q_\pm= c\D \pm  \calS\D,
\end{equation}
associated with the symbols
\begin{equation}
\hS(\xi) = \sqrt{
\frac{1+a\xi^2}{1+b\xi^2}}, 
\qquad
\hQ_\pm(\xi)= i\xi c \pm i\xi\hS(\xi),
\end{equation}
and observe 
\begin{equation}\label{e.sim}
\pmat{\calS & I\cr-\calS&I}
\pmat{\lambda-c\D & -\D\cr-\calS^2\D&\lambda-c\D} 
\pmat{\calS & I\cr-\calS&I}\inv
= \pmat{\lambda-\Q_+ &0\cr 0&\lambda-\Q_-}.
\end{equation}
The Fourier multipliers $\calS$, $\calS\inv$, and 
$(\lambda-\Q_\pm)\inv$ will be seen to be
bounded on $L^2_\ap$, uniformly for $\lambda$ of positive real part.
Lemma~\ref{l.lin1} concerning the basic properties of $\Lc$
is proved in section \ref{s.Lc}
(except for the proof of part (iv), which we provide in 
appendix~\ref{apx:Multzero}).
Before that we will develop necessary estimates for
various Fourier multipliers associated with the resolvent of $\BBo$.

\section{General Fourier symbol estimates}
Note that for any smooth $g\colon\R\to\R$ with compact support, 
by Plancherel's theorem, 
\[
\intR |e^{\ap x}g(x)|^2\,dx = \intR |\hat g(k+i\ap)|^2\,\frac{dk}{2\pi}.
\]
It follows that if $\RR$ is a Fourier multiplier operator with symbol $\hat\RR$
analytic and bounded on the strip where $0\le \Im\xi\le\ap$, 
then 
the operator norm of $\RR$ acting on $L^2_\ap$ is
\begin{equation}\label{b.Fmul}
\|\RR\|_\ap = \sup_{k\in\R} |\hat\RR(k+i\ap)| . 
\end{equation}

\begin{lemma} \label{l.sym}
Suppose $c>1$ and $0<\ap<\apc$ and $0<a<b$. 
For all real $k\ne0$, with $\xi=k+i\ap$ we have
\begin{eqnarray}
k\, \Im \hS(\xi)&<& 0, 
\label{hS1} \\
\sqrt{\frac ab}\ <\ |\hS(\xi)|&<& 
\hS(i\ap)= \sqrt{\frac{1-a\ap^2}{1-b\ap^2} } <c , 
\label{hS2a} \\ 
|\hS(\xi)|&<& 1 
-\frac12 \frac{(b-a)(k^2-\ap^2)} {1+b(k^2-\ap^2)} , 
\label{hS2} \\
i\xi\hS(\xi)&=& -\sqrt{-\xi^2 \frac{1+a\xi^2}{1+b\xi^2}}. 
\label{hS3}
\end{eqnarray} 
\end{lemma}
\begin{proof}
It suffices to consider $k>0$.
Taking $k\to\pm\infty$ and $k\to0$, note that due to (\ref{e.sol2}b),
\[
\hS(\pm \infty+i\ap) = \sqrt{\frac ab} <1< 
\sqrt{\frac{1-a\ap^2}{1-b\ap^2} } =\hS(i\ap) <c. 
\]
Observe that 
\begin{equation}
\hS(\xi)^2= 
\frac ab + \left(1-\frac ab\right) 
\frac{1}{1+b\xi^2}
= 
\frac ab + \left(1-\frac ab\right) 
\frac1{1+b(k^2-\ap^2)+2ibk\ap}.
\end{equation}
The last term has negative imaginary part and positive real part, 
which implies \qref{hS1} and the first part of \qref{hS2a}.
By the triangle inequality,
\begin{equation}
|\hS(\xi)^2| < 
\frac ab + \left(1-\frac ab\right) 
\frac1{1+b(k^2-\ap^2)} 
= \frac{1+a(k^2-\ap^2)}{1+b(k^2-\ap^2)} 
<  \frac{1-a\ap^2}{1-b\ap^2} .  
\end{equation}
This gives \qref{hS2a}, and
since $x\le \frac12+\frac12x^2$ for any $x\ge0$, 
taking $x=|\hS(\xi)|$ gives \qref{hS2}.

To prove \qref{hS3}, observe that since $\xi^2= k^2-\ap^2+2ik\ap$
and $k>0$, we have 
\[
0<\arg(1+b\xi^2)<\arg b\xi^2=\arg \xi^2<\pi,
\qquad 
0<\arg(1+a\xi^2)<\pi/2. 
\]
Hence the quantity
\[
\frac{-\xi^2}{1+b\xi^2} (1+a\xi^2)
\]
is never strictly negative, since its argument
is strictly between $-\pi$ and $\pi/2$.
Now \qref{hS3} follows by continuation starting at $k=0$.
\end{proof}

Since \qref{hS3} implies $i\xi\hS(\xi)$ has negative real part,
and since
\[
\Re\hQ_\pm(\xi) = 
-\ap c \pm \Re(i\xi\hS(\xi)) = -\ap c\mp(\ap\Re\hS(\xi)+k\Im\hS(\xi)),
\]
we infer the following by using \qref{hS2} and \qref{hS1} and \qref{hS3},
along with \qref{b.Fmul}.
\begin{corollary} \label{hQests}
For all real $k$, with $\xi=k+i\ap$ we have
\begin{align*}
& -2\ap c < \Re\hQ_+(\xi) < -\ap c,
\\
& -\ap c <  \Re\hQ_-(\xi) \le -\ap\left( c - 1+ \frac12 
\frac{(b-a)(k^2-\ap^2)} {1+b(k^2-\ap^2)} \right)<0. 
\end{align*}
Moreover, whenever $\Re\lambda+\ap(c-\hS(i\ap))\ge0$ we have
\begin{align*}
& \|(\lambda-\Q_+)\inv\|_\ap \le (\Re\lambda+\ap c)\inv, \\
& \|(\lambda-\Q_-)\inv\|_\ap \le (\Re\lambda+\ap (c -\hS(i\ap))\inv.
\end{align*}
\end{corollary}

\section{Basic properties of the linearization}
\label{s.Lc}

Here we provide the proof of Lemma~\ref{l.lin1}, parts (i)--(iii).
The proof of part (iv) appears in appendix~\ref{apx:Multzero}.

The proof of part (i) is straightforward:
Writing $\xi=k+i\ap$, the Fourier symbol of $B\inv \D^j$ satsifies
\[
\frac{(i\xi)^j}{1+b\xi^2} \to0 \quad\mbox{as $k\to\pm\infty$}
\]
for $j=0,1$. And for $g=q_c$, $r_c$, $q_c'$, $r_c'$, $g$ is continuous
with $g(x)\to0$ as $x\to\pm\infty$.
Hence the operators $B\inv\D^jg$ are all compact on $L^2_\ap$,  
by the convenient compactness criterion of \cite{P85},
using the isomorphism $g\to e^{\alpha x}g$ from $L^2_\ap$ to $L^2$.
Since $B\inv g\D=B\inv(\D g-g_x)$, the operator $\Lc-(\BBo)$ is a 
simple linear combination of compact operators, so is compact.

To prove part (ii), we observe that $\BBo$ is the generator
of a $C^0$ group on $L^2_\ap$.  After the change of variables
\qref{e.qr2}, this follows from the Hille-Yosida
theorem due to the resolvent bounds in Corollary~\ref{hQests}.
Now $\Lc$ generates a $C^0$ group on $L^2_\ap$ also,
by a standard perturbation theorem 
\cite[Ch. 3, Thm. 1.1]{Pa}.  

For part (iii) we note that 
the spectrum of $\BBo$ on $L^2_\ap$ is the
union of the image of the curves
\[
k\mapsto \lambda = \hQ_\pm(k+i\ap),
\]
which lie strictly in the left half-plane $\Re\lambda<0$ due
to Corollary~\ref{hQests}.
Then the essential spectrum of $\Lc$ lies in the left half-plane too,
by a standard generalization of Weyl's theorem to non-selfadjoint 
operators---One applies the analytic Fredholm theorem from 
\cite[I.5.1]{GK} or \cite[VI.14]{RSv1} 
in the right half-plane to the factorization
\[
I-(\lambda-L-c\D)\inv(\Lc-L-c\D)
= (\lambda-L-c\D)\inv (\lambda-\Lc) .
\]

\section{Reduction of the resolvent}

For simplicity we write $(q,r)=(q_c,r_c)$ henceforth.
The resolvent equation for the operator $\Lc$ 
takes the following form:
\begin{eqnarray}
(\lambda -c\D)q_1 -\D r_1 &=& f_1, \label{e.r1a}\\
(-A\D+r\D+2r')q_1+(B(\lambda-c\D)+q'+2q\D)r_1
&=& Bg_1.
\label{e.r1b}
\end{eqnarray}
We study this system in $L^2_\ap$, $0<\ap<\apc$, 
by changing variables using the transformation
\begin{equation} \label{e.qr2}
\begin{pmatrix}q_2\cr r_2\end{pmatrix} = 
\begin{pmatrix}\calS & I\cr -\calS & I\end{pmatrix}
\begin{pmatrix}q_1\cr r_1\end{pmatrix}  ,
\end{equation}
which is bounded on $L^2_\ap$ with bounded inverse,
due to \qref{hS2a}.
In the new variables the resolvent system \qref{e.r1a}-\qref{e.r1b} is written 
\begin{equation}\label{e.qr3}
\pmat{\lambda-\Q_+ &0\cr 0& \lambda - \Q_-}
\pmat{ q_2\cr r_2} 
+\pmat{ 
R_r & R_q \cr
R_r & R_q }
\pmat{I&-I\cr I&I}
\pmat{ q_2\cr r_2} 
= \pmat{ f_2\cr g_2 },
\end{equation}
with 
\begin{equation} \label{e.fg2}
\begin{pmatrix}f_2\cr g_2\end{pmatrix} = 
\begin{pmatrix}\calS & I\cr -\calS & I\end{pmatrix}
\begin{pmatrix}f_1\cr g_1\end{pmatrix}  ,
\end{equation}
\begin{equation} \label{d.RqRr}
R_r = \frac12 B\inv(r\D +2r')\calS\inv, \qquad
R_q = \frac12 B\inv(q'+2q\D ).
\end{equation}
Subtracting the second equation from the first, this system becomes
\begin{eqnarray}
\pmat{\lambda-\Q_+ & -\lambda+\Q_- \cr
R_q+R_r & \lambda-\Q_- +R_q- R_r }
\pmat{q_2\cr r_2} &=& 
\pmat{f_2-g_2\cr g_2}.
\end{eqnarray}
We can eliminate $q_2$ by writing
\begin{equation}\label{d.q2}
 q_2 = (\lambda-\Q_+)\inv (\lambda-\Q_-)r_2+ (\lambda-\Q_+)\inv(f_2-g_2), 
\end{equation}
and using this in the second equation. This reduces the resolvent
equation to the form
\begin{equation} \label{res.r2}
\boxed{\ww(\lambda) (\lambda-\Q_-)r_2 = g_3}
\end{equation}
with
\begin{equation}\label{d.Wlam}
\ww(\lambda)= I + (R_q+R_r)(\lambda-\Q_+)\inv + (R_q-R_r) (\lambda-\Q_-)\inv  ,
\end{equation}
and
\begin{equation}\label{d.g3}
g_3 = g_2-(R_q+R_r)(\lambda-\Q_+)\inv(f_2-g_2).
\end{equation}
Thus we see that to prove both Theorems~\ref{t.splin} and \ref{t.smlin} it will 
suffice to study the invertibility of the operator bundle $\ww(\lambda)$.
\begin{lemma}
If $\Re\lambda+\ap(c-\hS(i\ap))\ge 0$, then $\lambda$ is in the resolvent set of $\Lc$
if and only if $\ww(\lambda)$ is invertible.
\end{lemma}
For later use, note $R_q$ and $R_r$ are compact (since $B\inv \D q$
and $B\inv q'$ are compact), and
\begin{eqnarray}
R_q+R_r &=& B\inv \left(
-q'(\half+\half c\calS\inv) +\D q (1-\half c\calS\inv)
\right), \label{d.Rqpr}
\\
\label{d.Rqmr}
R_q-R_r &=& B\inv \left(
q'(\half+c\calS\inv) + q\D(1+\half c\calS\inv) 
\right) \\
\nonumber
&=&  \left(q'(\half+c\calS\inv) + q\D(1+\half c\calS\inv)\right) B\inv \\
&& \ +\ [B\inv,q']
(\half+c\calS\inv) + [B\inv,q]\D(1+\half c\calS\inv) ,
\nonumber
\end{eqnarray}
where 
\[
[B\inv,q']=B\inv q'-q'B\inv ,
\qquad
[B\inv,q]=B\inv q-qB\inv.
\]

\section{Spectral implies linear stability}
\label{s.splin}

In this section we complete the proof of Theorem~\ref{t.splin}. 
Fix $c>1$ and $\ap$ with $0<\ap<\apc$. 
By Lemma~\ref{l.lin1} and the hypothesis of Theorem~\ref{t.splin} concerning
eigenvalues, we know that the closed right half-plane 
is in the resolvent set of $\Lc$, except for the origin $\lambda=0$.
We will deduce the conclusion of the theorem by applying 
the Gearhart-Pr\"uss theorem in the spectral complement 
$Z_\ap=(I-P_c)L^2_\ap$ of the generalized eigenspace for $\lambda=0$.  
For this purpose,
it suffices to prove the uniform resolvent bound \qref{i.unif1}.
Due to the reduction carried out in the previous section,
to prove the uniform resolvent bound \qref{i.unif1}, it suffices to prove that 
\begin{equation}\label{lg.bd1}
\sup_{ \Re\lambda\ge0,\, |\lambda|>R}
\|(R_q\pm R_r)(\lambda-\Q_\pm)\inv\|_\ap \to 0
\qquad\mbox{as $R\to\infty$},
\end{equation}
since then $\ww(\lambda)\to I$ and $g_3$ is uniformly bounded 
in terms of $(f_2,g_2)$. 

From Corollary~\ref{hQests},
we know that $\lambda-\Q_\pm$ has bounded inverse whenever $\Re\lambda\ge0$.
Moreover, we claim that as $|\lambda|\to\infty$ with $\re\lambda\ge0$,
$(\lambda-\Q_\pm)\inv\to0$ in the strong operator sense on $L^2_\ap$.
To see this, fix $z\in L^2_{\ap}$, and let $w= (\lambda-\Q_\pm)\inv z$. Then
the Fourier transform 
\[
\hat w(k+i\ap) = (\lambda-\hQ_\pm(k+i\ap))\inv \hat z(k+i\ap)\to0
\]
for a.e. $k$. By dominated convergence it follows $\|w\|_{L^2_\ap}\to0$.

Since $R_q\pm R_r$ is compact, \qref{lg.bd1} follows
as a consequence of an abstract fact: In a Hilbert space, 
if a sequence of bounded operators $T_n$ converges strongly to $0$, 
and the operator $S$ is compact, then $ST_n$ converges to $0$ in operator norm.
(We omit the elementary proof.)

\section{Spectral stability in the KdV scaling regime}

Our goal in this section is to prove Theorem~\ref{t.smlin}, establishing
spectral stability for weakly nonlinear waves.
Our strategy involves making use of known stability properties 
of the soliton solution of the KdV equation
\begin{equation}\label{e.kdv}
\D_t \rho -\D_x \rho + 3\rho\D_x\rho + (b-a)\D_x^3\rho = 0,
\end{equation}
given by $\rho=\theta_0(x)$ where
\begin{equation}\label{d.Th}
\thez(x)=\sech^2\left(\frac12\hap_0 x\right), \qquad \hap_0=\frac1{\sqrt{b-a}}. 
\end{equation}
The eigenvalue problem for the linearization of \qref{e.kdv} about $\theta_0$
takes the form
\[
(\Lambda -\D_x + 3\D_x\theta_0 +(b-a)\D_x^3)\rho_1 = 0,
\]
which we rewrite as
\begin{equation}
\WW_0(\Lambda)(\Lambda-\D+(b-a)\D^3)\rho_1 = 0,
\end{equation}
in terms of the bundle 
\begin{equation}\label{d.W0}
\WW_0(\Lambda)= I+ (3\D\thez ) 
(\Lambda-\D+(b-a)\D^3)\inv . 
\end{equation}
Due to known stability properties of the KdV soliton 
(see Lemma~\ref{l.nullkdv} for a precise characterization),
$\WW_0(\Lambda)$ is known to be invertible in $L^2_\hap$ whenever 
$0<\hap<\hap_0$ and $\Lambda\ne0$ with $\Re\Lambda\ge -\hb$, 
where
\begin{equation}\label{d.hb}
\hb = \hap(1-(b-a)\hap^2).
\end{equation}
(The essential spectrum of $\D-(b-a)\D^3$ in $L^2_\hap$ 
is contained in the half-plane $\Re\Lambda\le-\hb$.)

To see the relevance of this KdV eigenvalue problem for small-energy 
solitary waves
of the Benney-Luke system \qref{eq:bousys}, we study the 
reduced eigenvalue problem from \qref{res.r2} using the KdV scaling,
\begin{equation} \label{kdvscale}
c=1+\frac{\ep^2}2, \quad 
\lambda=\frac{\ep^3}2 \Lambda , \quad
\hat x = \ep x.
\end{equation}
The solitary wave profile from \qref{e.sol2} then takes the form
\begin{equation}
q_c(x) = {\ep^2} \thep(\ep x), \quad \thep(\hat x)= \frac{c+1}{2c}
\sech^2\left(\frac12 \hbet \hat x\right),
\quad \hbet = \sqrt{\frac{c+1}{2(bc^2-a)}}.
\end{equation}

Formally, the KdV scaling corresponds to the following:
\[
\dx\sim\ep\dhx,  \quad \calS \sim I+\frac12(b-a)\ep^2\dhx^2, 
\quad q_c \sim -r_c \sim {\ep^2}\thez(\hat x).
\]
Using this scaling in the reduced resolvent equation 
\qref{res.r2} indicates
\[
R_q-R_r \sim 
\frac{3\ep^3}2 \left( \thez'+\thez\dhx \right)
= 
\frac{3\ep^3}2 \, \dhx\thez \ ,
\qquad
R_q+R_r \sim \frac{\ep^3}2 \left( 
-2\thez'+\dhx \thez
\right),
\]
\[
\lambda-\Q_- \sim \frac{\ep^3}2(\Lambda-\dhx+(b-a)\dhx^3),
\qquad
\lambda-\Q_+ \sim -2 \ep \dhx, 
\]
and consequently
\begin{equation} \label{l.W0}
\ww(\lambda) \sim 
I+ (3\dhx\thez ) (\Lambda-\dhx+(b-a)\dhx^3)\inv 
=\WW_0(\Lambda).
\end{equation}
A key step in the proof of Theorem~\ref{t.smlin} is to 
make the formal limit in \qref{l.W0} precise, and invoke
the operator-valued Rouch\'e theorem proved by  Gohberg and Sigal \cite{GS}
to deduce that $\ww(\lambda)$ is invertible for every nonzero $\lambda$ in
a suitable open half-space that contains the closed right half-plane.

We introduce a scaling operator defined by $\isc g(x)=g(\ep x)\sqrt\ep$ .
Then $\isc:L^2_{\hap}\to L^2_{\ep\hap}$ is an isometry, since
\[
\intR |e^{\hap x}g(x)|^2\,dx = 
\intR |e^{\ep \hap x}g(\ep x)\sqrt\ep|^2 \,dx .
\]
Then any bounded operator $Q$ on $L^2_{\ep\hap}$ induces a bounded operator 
$\isc\inv Q\isc$ on $L^2_{\hap}$ with the same norm.
We make \qref{l.W0} precise in the following sense: 
\begin{theorem}[Bundle convergence]\label{t.Bun} 
Fix $\hap\in(0,\hap_0)$ and let $\hb=\hap(1-(b-a)\hap^2)$. 
Then with $\wep(\Lambda)= \isc\inv \ww(\lambda)\isc$
where $\lambda=\frac12\ep^3\Lambda$, 
we have
\begin{equation}\label{l.W1}
\sup_{\Re\Lambda\ge -\hb/2}
\| \wep(\Lambda)-\WW_0(\Lambda)\|_\hap \to 0 \quad\mbox{as $\ep\to0$}.
\end{equation}
\end{theorem}

\subsection{Estimates for the bundle convergence theorem}
To prove Theorem~\ref{t.Bun} we substitute \qref{d.Rqpr}--\qref{d.Rqmr} 
into \qref{d.Wlam}.  The proof, to be completed in
subsection~\ref{ssec:bunpf}, follows from three groups of estimates
that we detail in this subsection:
(a) basic estimates on Fourier multipliers,
(b) convergence estimates for certain rescaled Fourier multipliers
in the KdV limit, and (c) estimates on junk terms and commutators.
We define
\begin{equation}\label{d.oms}
\Omega_* = \{\Lambda\in\C: \Re\Lambda\ge -\hb/2 \} .
\end{equation}

\begin{lemma}[Basic estimates]\label{lbun-a}
Uniformly for small $\ep>0$ and for $\Lambda\in\Omega_*$,
\begin{eqnarray}
\label{b.b1}  
\|\calS\inv\|_\ap &\le& C,
\\ \label{b.b2} 
\|B\inv\D^j\|_\ap  &\le& C \quad(j=0,1), 
\\ \label{b.b3} 
\|(\lambda-\Q_+)\inv\|_\ap  &\le& C\ep\inv,
\\ \label{b.b4} 
\| (\lambda-\Q_-)\inv\|_\ap  &\le& C\ep^{-3}.
\end{eqnarray}
\end{lemma}

\begin{lemma}[KdV limit of Fourier multipliers]\label{lbun-b}
For $j,k=0,1$,
\begin{equation}\label{l.F1}
\| \isc\inv 
\left( \ep^{3-j}\D^j\calS^{-k}B\inv(\lambda- \Q_-)\inv \right)\isc 
- 2 \D^j (\Lambda-\D+(b-a)\D^3)\inv \|_\hap \to 0,
\end{equation}
uniformly for $\Lambda\in\Omega_*$.
\end{lemma}
Lemma~\ref{lbun-b} is key, but its proof 
turns out not to be very hard, only involving Taylor expansion
of symbols at low frequency, and uniform bounds at high frequency.
The presence of the smoothing operator $B\inv$ simplifies the 
analysis as compared to the case of water waves treated in \cite{PS2}.

Finally, the junk terms include $(R_q+R_r)(\lambda-\Q_+)\inv$,
and terms involving the commutators $[B\inv,q']$ and
$[B\inv,q]$  in $R_q-R_r$.
The first kind of junk term is handled by noting
that since $q=O(\ep^2)$ and $q'=O(\ep^3)$,
the estimates \qref{b.b1}, \qref{b.b2} and \qref{b.b3} yield
\begin{equation}\label{e.Rqpr}
\|(R_q+R_r)(\lambda-\Q_+)\inv\|_\ap \le C\ep, 
\end{equation}
where $\ap=\ep\hap$ (here and below).
Concerning the commutators, we will establish the following. 
\begin{lemma}[Commutator estimates]\label{lbun-c}
\begin{equation}\label{e.Bcqx}
\| [B\inv,q']\|_\ap\le C \ep^4,
\end{equation}
\begin{equation}\label{e.keyB}
\| [B\inv,q]\D(\lambda-\Q_-)\inv\|_\ap \to 0 \qquad\mbox{as $\ep\to0$},
\end{equation}
uniformly for $\Lambda\in\Omega_*$.
\end{lemma}
The first of these estimates is 
not difficult.  However, it turns out that the 
term $[B\inv,q]\D$ has operator norm
$\|[B\inv,q]\D\|_\ap = O(\ep^3)$, which is not small enough to neglect,
due to \qref{b.b4}.  
Consequently, we have 
to establish instead the more complicated commutator estimate
in \qref{e.keyB}.
To establish this we will use the commutator estimate 
in Lemma \ref{L.comm} below (from \cite{PS2}), and deal separately
with the high and the low frequencies.

The result of Theorem~\ref{t.Bun} follows directly 
from the symbol limits in \qref{l.F1} and the estimates in \qref{e.Rqpr},
\qref{e.Bcqx}, \qref{e.keyB} and \qref{l.F1},
using \qref{d.Rqmr} and the fact that $\isc$ is an isometry from
$L^2_\hap$ to $L^2_\ap$.


\subsection{Basic estimates}
We now prove Lemma~\ref{lbun-a}.
From \qref{hS2a} and \qref{b.Fmul} we infer
\begin{equation}\label{b.b1p}
\|\calS\inv\|_\ap = \sup_{k\in\R} |\hS(k+i\ap)\inv| = \sqrt{\frac ba}.
\end{equation}
Moreover, 
\begin{equation}\label{b.b2pa}
\|B\inv\|_\ap =\sup_{k\in \R} \left| \frac1{1+b(k^2-\ap^2)+2ibk\ap}\right| =
\frac1{1-b\ap^2},
\end{equation}
and since $2|\xi|\le 1+|\xi|^2$, 
\begin{equation}\label{b.b2pb}
\|B\inv\D\|_\ap =\sup_{k\in \R} \left| \frac{k+i\ap}
{1+b(k^2-\ap^2)+2ibk\ap}\right|
\le \sup_{k\in\R} \frac12 \left| \frac{1+k^2+\ap^2} 
{1+b(k^2-\ap^2)}\right| \le C.
\end{equation}
Next we invoke Corollary~\ref{hQests} 
with $\Lambda\in\Omst$ and $\lambda=\half\ep^3\Lambda$. 
Since $\hb\le\hap$ we have $\Re\Lambda+\hap\ge0$, so
\[
\Re\lambda+\ap c = 
\half \ep^3 \Re\Lambda +\ep\hap(1+\half\ep^2) \ge \ep\hap,
\]
hence
\begin{equation}\label{b.b3p}
\|(\lambda-\Q_+)\inv\|_\ap \le \frac1{\ep\hap}.
\end{equation}
Also, since
\[
\hS(i\ap)=\sqrt{\frac{1-a\ap^2}{1-b\ap^2}} \le 1+\frac12(b-a)\ep^2\hap^2,
\]
we have
\begin{eqnarray*}
\Re\lambda+\ap(c-\hS(i\ap)) &\ge& 
\frac{\ep^3}2 (\Re\Lambda  + \hap (1-(b-a)\hap^2)) \ge \frac{\ep^3\hb}{4},
\end{eqnarray*}
hence
\begin{equation}\label{b.b4p}
\|(\lambda-\Q_-)\inv\|_\ap \le \frac4{\ep^3\hb}.
\end{equation}

\subsection{KdV limit of Fourier multipliers}
Next we prove Lemma~\ref{lbun-b}.
By \qref{b.Fmul}, this is equivalent to showing that for $j,\hat j=0,1$,
\begin{equation}\label{l.F2}
\left| \frac{\hS(\ep\hxi)^{-\hat j}}{1+b\ep^2\hxi^2} 
\frac{\ep^3(i\hxi)^j}{\lambda-\hat\Q_-(\ep\hxi)} 
- \frac{ 2 (i\hxi)^j}{\Lambda-i\hxi-(b-a)i\hxi^3}
\right| \to 0  \quad\mbox{as $\ep\to0$},
\end{equation}
uniformly for $\hxi=\hat k+i\hap$ with $\hat k\in \R$, 
and uniformly for $\Lambda\in\Omst$, with $\lambda=\half\ep^3\Lambda$.
The factor in \qref{l.F2} that corresponds to the symbol of $\calS\inv B\inv$ satisfies
\begin{equation}\label{l.hSe}
\frac{\hS(\ep\hxi)^{-\hat j}}{1+b\ep^2\hxi^2}  = 1+O(\ep\hxi),
\end{equation}
and is uniformly bounded. To establish \qref{l.F2}, we will treat separately
the low and high frequencies.  

To start, we obtain a basic lower bound on the denominator of
the second term in \qref{l.F2},
\begin{equation}\label{b.mo}
\mo= {\Lambda-i\hxi-(b-a)i\hxi^3}.
\end{equation} 
Observe that for $\hxi=\hk+i\hap$ and $\Lambda\in\Omst$
we have the estimate
\begin{equation}\label{b.lomo}
\Re \mo = \Re\Lambda + \hap+(b-a)(3k^2-\hap^2)\hap 
\ge
\frac{\hb}2 + 3(b-a)\hap k^2 \ge \frac{\hap |\hxi|^2}C,
\end{equation}
provided $\half\hb C\ge \hap^3$ and $3(b-a)C\ge1$.

\subsubsection{Low frequency (KdV) regime: $|\ep\hxi|\le 4\ep^p$.}
We fix $p\in(\frac13,\frac12)$, and let
\[
I_0=\{\hxi=\hat k+i\hap: \hat k\in \R \mbox{\ and\ } |\ep\hxi|\le 4\ep^p\}.
\]
For frequencies in this regime we carry out a Taylor expansion
of the symbols in \qref{l.F2}, handling the remainder carefully.
Observe that
\[
\hS(\xi)^2 = 
\frac{1+a\xi^2}{1+b\xi^2}=1-(b-a)\xi^2+O(\xi^4), 
\]
so
\begin{equation}\label{ex.S}
\hS(\ep\hxi)= 1-(b-a)\frac{\ep^2\hxi^2}2 +O(\ep^4\hxi^4).
\end{equation}
Hence 
\begin{eqnarray}
\lambda-\hat\Q_-(\ep\hxi) &=& 
\frac{\ep^3}{2}\Lambda-i\ep\hxi(1+\frac{\ep^2}2) + i\ep\hxi \hS(\ep\hxi) 
\nonumber \\
&=& \frac{\ep^3}2\left( 
\Lambda-i\hxi-(b-a)i\hxi^3 + \hxi^3O(\ep^2\hxi^2) \right).
\label{lowf1} 
\end{eqnarray}

Let us define
\begin{equation}
\mep= 2\ep^{-3}(\lambda-\hQ_-(\ep\hxi)).
\end{equation}
Then by \qref{lowf1} we have that 
\[
E:= \mep-\mo = \hxi^3 O(\ep^2\hxi^2) =\hxi^3 O(\ep^{2p}) .
\]
Then due to the lower bound \qref{b.lomo}, for $\hxi\in I_0$ we have
\begin{equation}\label{i.Emo}
\left|\frac{E}{\mo}\right| \le \frac{C |\hxi^3| \ep^{2p}}{|\mo|} 
\le C|\hxi|\ep^{2p}\le C\ep^{3p-1},
\end{equation}
which tends to zero as $\ep\to0$. Then it follows from  \qref{i.Emo}
and \qref{b.lomo} that
\begin{equation}
\left| \frac{(i\hxi)^j}{\mep}-
 \frac{(i\hxi)^j}{\mo}\right| = \frac{ |\hxi|^j |E/\mo|}{|\mo||1+E/\mo|} 
\le C|\hxi|^{j-2}\ep^{3p-1} \le C \ep^{3p-1}
\end{equation}
and consequently \qref{l.F2} holds 
uniformly for $\hxi\in I_0$ and $\Lambda\in\Omst$.

\subsubsection{High frequency regime: $|\ep\hk|\ge 2\ep^p$. }
Consider $\hxi$ in the set 
\[
I_1=\{ \hxi=\hk+i\hap: |\ep\hk|\ge 2\ep^p\},
\]
and note that we have $I_0\cup I_1=\R+i\hap$ for sufficiently small
$\ep>0$. 
In this complementary regime we claim that the terms in \qref{l.F2} 
separately go to zero. Consider the second term first.
From the lower bound \qref{b.lomo}, we find that this term is bounded by
\begin{equation}\label{i.mzero}
\left| \frac{(i\hxi)^j}{\mo}\right| \le C|\hxi|\inv \le C\ep^{1-p} \to0.
\end{equation}
Now consider the first term in \qref{l.F2}.
With $\xi=k+i\ap=\ep\hxi$, for small enough $\ep$ we have 
$k^2-\ap^2>\frac12 k^2\ge 2\ep^{2p}$ and 
\[
\frac{b (k^2-\ap^2)} {1+b(k^2-\ap^2)} \ge 
\frac{ 2b \ep^{2p}} {1+2b\ep^{2p}}\ge b\ep^{2p}.
\]
By Corollary~\ref{hQests}, since $c-1=\frac12\ep^2$ and
$\Re\Lambda+\hap\ge0$ we then get
\begin{equation}
\Re\left(\frac{\ep^3}2\Lambda- \hQ_-(\ep\hxi)\right) 
\ge 
\frac{\ep^3}2 \Re\Lambda +\frac{\ep\hap}2
\left(\ep^2+(b-a)\ep^{2p}\right)
\ge 
\frac{\hap}2 (b-a) \ep^{1+2p}.
\end{equation}
By consequence we have that for sufficiently small $\ep>0$, 
\begin{equation}\label{l.fhi1}
\left| \frac {\ep^2}{\lambda-\hQ_-(\ep\hxi)} \right| \le C \ep^{1-2p} \to 0.
\end{equation}
Since by \qref{b.b2pa}-\qref{b.b2pb} we have
\begin{equation}\label{b.Bb}
\left| \frac{ \ep(i\hxi)^j}{1+b\ep^2\hxi^2}\right| \le C
\end{equation}
for $j=0,1$,
we see that the first term in \qref{l.F2} tends to zero,
uniformly for $\hxi\in I_1$ and $\Lambda\in\Omst$. 

This finishes the proof of the limit formula \qref{l.F1} for Fourier mulitpliers.

\subsection{Commutator estimates}
In this subsection we prove Lemma~\ref{lbun-c}.
The proof of the following commutator bounds, from \cite{PS2}, is short and is
reproduced here for completeness.
We write $\ip{k}=(1+|k|^2)^{1/2}$ below.

\begin{lemma}\label{L.comm}
Let $\PP$, $\QQ$ and $\RR$ be Fourier multipliers
with symbols $\hat\PP$, $\hat\QQ$ and $\hat\RR$ respectively, and  let $s\ge0$.
Let $g(x)=\ep^2 G(\ep x)$ where $G\colon\R\to\R$ is smooth and exponentially
decaying, and let $h\colon\R\to\R$ be smooth with compact support.
Then
\[
\| \PP[\QQ,g]\RR h\|_{L^2} \le M_\ep M_G \|h\|_{L^2},
\]
where
\[
M_\ep = \sup_{k,\hat k\in\R} \ep^2 \frac{
\hat\PP(\ep k) |\hat\QQ(\ep k)-\hat\QQ(\ep\hat k)| \hat\RR(\ep \hat k)
}{\ip{k-\hat k}^s},
\qquad 
M_G = \intR \ip{k}^s|\hat G(k)|\,\frac{dk}{2\pi}.
\]
\end{lemma}

\noindent{\it Proof.} Using the Fourier transform and Young's
inequality, since $\hat g(k)=\ep\hat G(k/\ep)$, we have
\begin{eqnarray*}
\| \PP[\QQ,g]\RR h\|_{L^2}^2 &=& 
\intR \left| \intR \hat\PP(k) (\hat\QQ(k)-\hat\QQ(\hat k))\ep 
\hat G\left(
\frac{k-\hat k}{\ep} 
\right)
\hat\RR(\hat k) 
\hat h(\hat k)\,\frac{d\hat k}{2\pi}\right|^2 \frac{dk}{2\pi}
\\
&\le& M_\ep^2 \intR 
\left(\intR 
\ipbig{ \frac{k-\hat k}{\ep} }^s
\left|\hat G\left(
\frac{k-\hat k}{\ep} 
\right)\right| |\hat h(\hat k)|\,
\frac{d\hat k}{2\pi\ep}\right)^2\frac{dk}{2\pi}
\\
&\le& M_\ep^2 M_G^2 \|h\|_{L^2}^2.
\end{eqnarray*}

\subsubsection{Main commutator estimate.} Recall the key estimate 
\qref{e.keyB} that we need is
\begin{equation}\label{e.keyB2}
\| [B\inv,q]\D(\lambda-\Q_-)\inv\|_\ap \to 0 \qquad\mbox{as $\ep\to0$}
\end{equation}
in the $L^2_\ap$ operator norm. We apply the Lemma with $g=q$ 
so $G=\thep$ and $M_G=O(1)$, and take the symbols
\[
\hat\PP(k)=1, \qquad \hat\QQ(k) = 
\frac1{1+b(k+i\ap)^2},
\qquad \hat\RR(k)= \frac{i(k+i\ap)}{\lambda-\hQ_-(k+i\ap)},
\]
with $\ap=\ep\hap$. Taking any $s\ge1$ should work.
Then, writing $\xi=k+i\hap$, $\hat\xi=\hat k+i\hap$, 
since $(\ep\xi)^2-(\ep\hat\xi)^2=\ep(k-\hat k)(\ep\xi+\ep\hat \xi)$, we find
\begin{eqnarray*}
M_\ep &=& 
\sup_{k,\hat k\in\R} 
\frac{\ep^2}{\ip{k-\hat k}^s}
\left| 
\frac{1}{1+b\ep^2\xi^2} -
\frac{1}{1+b\ep^2\hat \xi^2} 
\right| |\hat\RR(\ep\hat k)|
\\ &\le&
\sup_{k,\hat k\in\R} 
\frac{ b|\ep\xi+\ep\hat\xi|}
{ |1+b\ep^2\xi^2| |1+b\ep^2\hat\xi^2| }
\frac{\ep^3 |\ep\hat\xi| }{|\lambda-\hQ_-(\ep\hat\xi)|}
\\ &\le&
C\sup_{\hat k\in\R} 
\frac{ (1+b|\ep\hat\xi|)|\ep\hat\xi|}
{|1+b\ep^2\hat\xi^2| }
\frac{\ep^3}{|\lambda-\hQ_-(\ep\hat\xi)|}
\end{eqnarray*}
Here we used the bound \qref{b.Bb} that follows from 
\qref{b.b2pa}-\qref{b.b2pb}.
We now treat separately the low and high frequency regimes.
In the low frequency regime $|\ep\hat\xi|\le 4\ep^p$ we get 
the bounds 
\begin{equation}
\frac{ (1+b|\ep\hat\xi|)|\ep\hat\xi|}
{|1+b\ep^2\hat\xi^2| }\le C\ep^p, \qquad
\frac{\ep^3}{|\lambda-\hQ_-(\ep\hat\xi)|} \le C,
\end{equation}
and in the high-frequency regime $|\ep\hat k|\ge 2\ep^p$ we have 
\begin{equation}
\frac{ (1+b|\ep\hat\xi|)|\ep\hat\xi|}
{|1+b\ep^2\hat\xi^2| }\le C, \qquad
\frac{\ep^3}{|\lambda-\hQ_-(\ep\hat\xi)|} \le C \frac{\ep^3}{\ep^{1+2p}}
=C\ep^{2-2p} ,
\end{equation}
Consequently $M_\ep\to0$ as $\ep\to0$, proving \qref{e.keyB2}.

\subsubsection{Simple commutator estimate.} In order to prove 
\begin{equation}\label{b.s3p}
\|[B\inv,q']\|_\ap \le C\ep^4,
\end{equation}
we take $g(x)=q'(x)=\ep^3\thep(\ep x)$, so $G(x)=\ep\thep(x)$ and $M_G\le C\ep$,
and take $\hat\QQ(k)$ as above,  and $\hat\PP(k)=\hat\RR(k)=1$.  Then the Lemma now yields
\begin{eqnarray*}
M_\ep \le 
\sup_{k,\hat k\in\R} 
\frac{ \ep^3 b|\ep\xi+\ep\hat\xi|}
{ |1+b\ep^2\xi^2| |1+b\ep^2\hat\xi^2| }
\le C\ep^3,
\end{eqnarray*}
whence \qref{b.s3p} follows since $M_\ep M_G\le C\ep^4$.

This finishes the proof of the bundle convergence theorem \ref{t.Bun}.

\subsection{Proof of Theorem~\ref{t.smlin}}\label{ssec:bunpf}
\begin{lemma} \label{l.kdvlim}
$\|\WW_0(\Lambda)-I\|_\hap \to 0$ 
as $|\Lambda|\to\infty$ with $\Re\Lambda\ge-\hb/2$.
\end{lemma}
\begin{proof}
This follows from the estimate \qref{i.mzero} for $|\ep\hk|\ge 2\ep^p$,
together with the estimate
\[
\left| \frac{(i\hxi)^j}{\mo}\right| \le \frac{C\ep\inv}
{|\Lambda|-C\ep^{-3}} \le C\ep^{1-p}
\]
for $|\ep\hk|\le 1$, $j=0,1$ and
for $|\Lambda|$ sufficiently large depending on $\ep$.
\end{proof}

As a consequence of Lemma \ref{l.kdvlim},
there exists $M_0>0$ such that for $\ep>0$ sufficiently small, 
$\|\WW_0(\Lambda)-I\|_\hap <\frac14$. 
Applying the bundle convergence theorem \ref{t.Bun}, we infer that
for small enough $\ep$,  $\wep(\Lambda)$ is invertible 
for $\Re\Lambda\ge -\hb/2$ and $|\Lambda|\ge M_0$.
This implies $\Lc$ has no eigenvalue satisfying
$\Re\lambda\ge -\frac14\ep^3\hb$ and $|\lambda|\ge\frac12\ep^3M_0$.

Moreover, with
$\hat\Omega =\{\Lambda: |\Lambda|\le M_0,\, \Re\Lambda\ge -\hb/2\}$,
then for small enough $\ep>0$, 
\begin{equation}
\|(\wep(\Lambda)-\WW_0(\Lambda))\WW_0(\Lambda)\inv\|_\hap<1
\end{equation}
for all $\Lambda\in \D\hat\Omega$. 
By the Rouch\'e theorem of Gohberg and Sigal \cite{GS}, it follows that
the total null multiplicity of characteristic values of the bundle $\wep(\Lambda)$
for $\Lambda\in \hat\Omega$ agrees with that of $\WW_0(\Lambda)$.
Denoting these multiplicities respectively by $m(\hat\Omega,\wep)$
and $m(\hat\Omega,\WW_0)$, we have 
\begin{equation}
m(\hat\Omega,\wep)
= 
m(\hat\Omega,\WW_0).
\end{equation}
As discussed in Appendix C, the null multiplicity of the characteristic value
$0$ is at least 2 for $\ww$, and $m(\hat\Omega,\WW_0)\le 2$.  
Hence $m(\hat\Omega,\wep)=2$, so $\Lambda=0$ is the only characteristic value
of $\wep$ in $\hat\Omega$.

This implies that for all nonzero $\lambda$
satisfying $\Re\lambda\ge -\frac14\ep^3\hb$,
$\ww(\lambda)$ is invertible and so $\lambda$ is not an eigenvalue of $\Lc$.
This concludes the proof of Theorem~\ref{t.smlin}.


%% file: mpq-part2.tex
\section{Decomposition of perturbed solitary waves}
In this part we prove Theorems \ref{thm:1} and \ref{thm:2}.
Let 
\[
\mathcal{M}=\{u_c(\cdot-x_0)\mid c^2>1,\ x_0\in\R\}
\] 
denote the two-dimensional manifold of solitary-wave states for the
Benney-Luke system \qref{eq:bousys}. 
To describe the behavior of solutions near $\mathcal{M}$, 
we will represent them using the ansatz
\begin{equation}
  \label{eq:decomp1}
u(t,x)=u_{c(t)}(y)+v(t,y), \qquad y=x-x(t).
\end{equation}
Here $u_{c(t)}$ comprises the main solitary-wave part of the solution and
$v$ is a remainder. The modulating parameters $c(t)$ and $x(t)$ 
describe the speed and phase of the main solitary wave at time $t$.
Substituting \eqref{eq:decomp1} into \eqref{eq:bousys}
and noting $cu_c'+Lu_c+f(u_c)=0$, we require
\begin{equation}
  \label{eq:v}
  \pd_t v=\mathcal{L}_{c(t)}v+(\dot{x}(t)-c(t))\pd_yv+l(t)+f(v),
\end{equation}
where
$\mathcal{L}_c=L+c\pd_y+f'(u_c)$ and $\dot x=dx/dt$ and
\begin{align*}
& l(t)=
(\dot{x}(t)-c(t))\pd_yu_{c(t)}(y)
-\dot{c}(t)\pd_cu_{c(t)}(y).
\end{align*}

If we were only going to consider initial data that is exponentially
well-localized, we could impose the nonsecularity condition
$P_{c(t)}v(t)=0$ at this point and study \qref{eq:v} in an
exponentially weighted space $H^1_\ap$, using the exponential decay
estimate supplied by Lemma~\ref{lem:lineardecay2}.  However, this is not
feasible for arbitrary small-energy perturbations of solitary waves. The
reason is that the spectral projection $P_c$ is not continuous on the
energy space $H^1$, due to the fact that an element of the
generalized kernel of the adjoint $\Lc^*$ does not decay as $x\to\infty$.

To deal with this difficulty, as in \cite{Miz09}
we split the remainder $v(t)$ into a part generated by free
propagation from the initial perturbation, and a well-localized part 
arising from interaction with the main solitary wave. We write
\begin{equation}
  \label{eq:decomp2}
v(t,y)=v_1(t,x)+v_2(t,y),
\end{equation}
where $v_1(t,x)$ is the solution to
\begin{equation}
  \label{eq:v1}
\left\{
  \begin{aligned}
& \pd_tv_1=Lv_1+f(v_1)\quad\text{for $(t,x)\in\R^2$},\\
& v_1(0,x)=v_0(x)\quad\text{for $x\in\R$.}    
  \end{aligned}\right.
\end{equation}
The freely propagating perturbation $v_1$ will decay locally in a coordinate 
frame following the main solitary wave, due to the viral estimates that
we establish in section~\ref{sec:virial}.  The remainder $v_2$ satisfies
\begin{equation}
  \label{eq:v2}
  \left\{\begin{aligned}
& \pd_tv_2=\mathcal{L}_{c(t)}v_2+(\dot{x}-c)\pd_yv_2+l+\kone+\ktwo,\\
& v_2(0,y)=0,
    \end{aligned}\right.
\end{equation}
where
\begin{equation}
  \label{eq:ks}
\kone=f'(u_{c(t)})\ti v_1(t),
\qquad
\ktwo=f(v(t))-f(\ti v_1(t)),
\qquad \ti v_1(t,y)=v_1(t,y+x(t)).
\end{equation}
This part will be `slaved' to $v_1$ via the estimates in exponentially
weighted norm that are provided in 
Lemma~\ref{lem:lineardecay2}.
To enable the use of that Lemma and fix the decomposition,
we will impose the constraint $P_{c(t)}v_2(t)=0$. 
In terms of the elements $\z_{1,c}^*$, $\z_{2,c}^*$ 
described in Appendix \ref{apx:Multzero}, 
that span the generalized kernel of $\Lc^*$,
this means
\begin{equation}
  \label{eq:orth}
\la v_2(t),\z^*_{1,c(t)}\ra=0,\quad 
\la v_2(t),\z^*_{2,c(t)}\ra=0\,.
\end{equation}

{\em Notation.}
Some additional notation to be used in Part II is as follows.
We will write $g\lesssim h$ to mean that 
there exists a positive constant such that $g\le Ch$. 
For $\R^2$-valued functions $g=(g_1,g_2)$ and $h=(h_1,h_2)$ let 
\[
\la g,h\ra=\int_\R (g_1(x)h_1(x)+g_2(x)h_2(x))\,dx.
\] 

For a Banach space $X$ 
we denote by $B(X)$ the space of 
all continuous linear operators on $X$.

\section{Local existence and continuation of the decomposition}
In this section we establish the validity of the representation
described above in \qref{eq:decomp1}--\qref{eq:orth}.
We first show that $u(t)-v_1(t)$ remains in $H^1_\alpha$ whenever
$0\le \alpha<\alpha_{c_0}.$  
Recall $\apc$ from \qref{e.sol2} is the exponential decay rate of the
wave profile $u_c$, and $\apc<b^{-1/2}$.
\begin{lemma}
  \label{lem:uv1}
Let $c_0>1$, $x_0\in\R$ and $v_0\in H^1$. Let $u(t)$ be a solution to
\eqref{eq:bousys} satisfying $u(0)=u_{c_0}(\cdot-x_0)+v_0$ and let 
$v_1$  be a solution to \eqref{eq:v1}. Then 
for every $\alpha\in(-\alpha_{c_0},\alpha_{c_0})$, 
\begin{equation}
  \label{eq:u-v1}
u(t)-v_1(t)\in 
C([0,\infty);H^1_\alpha(\R;\R^2))\cap C^1([0,\infty);L^2_\alpha(\R;\R^2)).
\end{equation}
\end{lemma}
\begin{proof}
By standard well-posedness arguments, $u$, $v_1$, and 
$w=u-v_1$ lie in $C(\R;H^1)$.
Writing
\[
v_1= \pmat{q_1\\ r_1}, 
\quad 
w=\pmat{\tilde{q}\\ \tilde{r}},
\]
we find $w$ satisfies a linear equation
\begin{equation}
  \label{eq:w}
\left\{  \begin{aligned}
& \D_tw=Lw+F(t)w,\\
& w(0)=u_{c_0}(\cdot-x_0),    
  \end{aligned}\right.
\end{equation}
where
\[
F(t)w=-B\inv\begin{pmatrix}0 \\ 
\D_x(r\ti q+2q\ti r)+ \ti q\D_x(2r_1-r)+\ti r\D_x(q_1-2q)
\end{pmatrix}.
\]
Since $B\inv$ and $B\inv\D_x$ are bounded on $L^2_\ap$,
we have 
\[
\|F(t)w(t)\|_{H^1_\alpha} \lesssim
(1+\|u(t)\|_{H^1}+\|v_1(t)\|_{H^1})\|w\|_{H^1_\alpha},
\]
and $F(t)\in C(\R;B(H^1_\alpha))$. 
Since $e^{tL}$ is a $C^0$-semigroup on 
both spaces $H^1$ and $H^1\cap H^1_\alpha$, and $u_{c_0}$ lies there,
it follows that 
\eqref{eq:w} has a solution in $C([0,\infty);H^1\cap H^1_\alpha)$
which agrees with $u-v_1$ by uniqueness in $H^1$.
This proves \eqref{eq:u-v1}. 
\end{proof}

Next, we associate a unique phase/speed pair to each $u$ near
$u_{c_0}$ in $L^2_\ap$.
Here and below we will make use of the following pointwise estimates 
for the neutral and adjoint neutral modes 
$\z_{j,c}$, $\z_{j,c}^*$,
satisfied uniformly for $c$ in a neighborhood of $c_0$:
\begin{gather} 
|\zone| +|\zatwo|
\lesssim e^{-\alpha_c|y|},
\qquad |\ztwo| + |\D_c\zatwo| 
\lesssim 
e^{-\alpha_c|y|}
(1+|y|) ,
\label{i:adj1}\\
|\zaone| \lesssim 
\min(1, e^{\alpha_c y}(1+|y|) ),
\qquad
|\D_c\zaone| 
\lesssim 
\min(1, e^{\alpha_c y}(1+|y|^2) ).
\label{i:adj2}
\end{gather}

\begin{lemma}
 \label{lem:decomp}
Let $c_0>1$ and $\alpha\in(0,\alpha_{c_0})$. 
Then there exist positive constants $\delta_0$, $\delta_1$ such that
with 
\[
U_0=\{w\in L^2_\alpha: 
\|w-u_{c_0}\|_{L^2_\alpha}<\delta_0\},
\quad
U_1=\{(\gamma,c)\in\R^2:|\gamma|+|c-c_0|<\delta_1\},
\]
then for each $w\in U_0$ there is a unique $(\gamma,c)\in U_1$ satisfying
\begin{gather*}
\la w(\cdot+\gamma)-u_c,\z_{1,c}^*\ra
=\la w(\cdot+\gamma)-u_c,\z_{2,c}^*\ra=0\,.
\end{gather*}
Further, the mapping $w\mapsto \Phi(w)=(\gamma,c)$ is smooth.
\end{lemma}
\begin{proof} 
The map $G:L^2_\ap \times\R \times (0,\infty) \to\R^2$ defined by
\begin{gather}
 G(w,\gamma,c)=
\pmat{
\ip{w-u_c(\cdot-\gamma),\zaone(\cdot-\gamma)} \\
\ip{w-u_c(\cdot-\gamma),\zatwo(\cdot-\gamma)}
}
\end{gather}
is smooth since $(\gamma,c)\mapsto \z_{j,c}^*(\cdot-\gamma)$ is
smooth with values in $L^2_{-\ap}=(L^2_\ap)^*$, due to 
the definitions in \qref{d:zz} and Lemma~\ref{l:adj}.
Moreover, $G(u_c,0,c)=0$ and 
\[
\frac{\pd G}{\pd(\gamma,c)}
(u_c,0,c) = 
\pmat{
\ip{\D_y u_c ,\zaone} & \ip{-\D_c u_c ,\zaone}
\\ \ip{\D_y u_c ,\zatwo} &\ip{-\D_c u_c ,\zatwo}
}
= \pmat{1&0\\0&1},
\]
due to \qref{e:zaeq}.
Thus the result follows immediately from the implicit function theorem. 
\end{proof}

Now we establish the local existence of the desired representation of
solutions, and we provide a continuation principle that ensures its existence
as long as a suitable distance to $\mathcal{M}$ and the wave-speed variation
remain small.  Since the manifold $\mathcal{M}$ is translation
invariant, we need only to use the local
coordinates in Lemma~\ref{lem:decomp}, without needing to study the global
geometry of $\mathcal{M}$ as in \cite{FP2}.

\begin{proposition}\label{cor:decomp}
Make the assumptions of Lemma~\ref{lem:uv1}, let $0<\ap<\ap_{c_0}$,
and let $\delta_0$, $\delta_1$ be given by Lemma~\ref{lem:decomp}.
Then there exist $T>0$ and $C^1$ functions $x(t)$, $c(t)$ on $[0,T)$
satisfying 
\begin{equation}\label{eq:xc0}
x(0)=x_0,\quad c(0)=c_0,\quad |c(t)-c_0|<\delta_1,
\end{equation}
such that if $v_2$ is defined by the decomposition 
\begin{equation}\label{e:udecom}
u(t,x)=u_{c(t)}(y)+ v_1(t,x)+v_2(t,y), \qquad y=x-x(t),
\end{equation}
then the orthogonality relations \qref{eq:orth} hold for all $t\in[0,T)$. 

Moreover, if $T<\infty$ and 
\begin{equation}\label{eq:supv}
\sup_{t\in[0,T)} 
\| u_{c(t)} + v_2(t) -u_{c_0}\|_{L^2_\ap} <\delta_0,
\end{equation}
then $T$ is not maximal.
\end{proposition}

\begin{proof}
Define $w(t;\hat x)=(u-v_1)(t,\cdot+\hat x)$ for $\hat x\in\R$.
Since $w(0;x_0)=u_{c_0} \in U_0$ by assumption, 
there exists $T_1>0$ such that 
$t\mapsto w(t;x_0)$ is $C^1$ with values in $U_0$ for $t\in[0,T_1)$.  
Then $(\gamma(t), c(t)):=\Phi(w(t;x_0))$ are the unique points in
$U_1$ such that $G(w(t;x_0+\gamma),0,c)=0$.
It follows that with $x(t):=x_0+\gamma(t)$, 
and with $v_2(t)=w(t;x(t))-u_{c(t)}$ given  by \qref{e:udecom},
\qref{eq:xc0} and \qref{eq:orth} hold for $t\in[0,T_1)$.
Moreover $\gamma(t)$ and $c(t)$ are $C^1$ 
because $\Phi$ is $C^1$ on $U_0$.

Suppose now that $C^1$ functions $x(t)$, $c(t)$ exist on $[0,T)$ 
such that \qref{eq:xc0} and \qref{eq:orth} 
hold with $v_2$ given by \qref{e:udecom},
which means such that for $0\le t<T$,  we have \qref{eq:xc0} and
\begin{equation}\label{e:G1}
 G(w(t;x(t)),0,c(t))=0.
\end{equation}
Suppose further that \qref{eq:supv} holds. Then
there is a closed ball $\hat U_0\subset U_0$
such that for all $t\in[0,T)$, $w(t;x(t))\in\hat U_0$.
Since $w\colon[0,T+1]\to L^2_\ap$ is uniformly continuous,
by enlarging $\hat U_0$ if necessary we can say 
there exists $\tau_0>0$ such that
whenever $\hat t\in[0,T)$ and $\tau\in[0,2\tau_0]$,
\begin{equation}\label{e:wU0}
w(\hat t+\tau;x(\hat t)) \in \hat U_0. 
\end{equation}
Fix $\hat t=T-\tau_0$. Applying Lemma~\ref{lem:decomp}, 
we infer that 
$(\hat\gamma(\tau), \hat c(\tau)):=\Phi(w(\hat t+\tau;x(\hat t)))$ 
are the unique points in $U_1$ such that 
for $\tau\in[0,2\tau_0]$,
\begin{equation}\label{e:G2}
G(w(\hat t+\tau;x(\hat t)+\hat\gamma),0,\hat c)=0.
\end{equation}
Also, $\hat\gamma(\tau)$ and $\hat c(\tau)$ are $C^1$, and the 
values $(\hat\gamma(\tau),\hat c(\tau))$ lie in a compact 
$\hat U_1\subset U_1$ for $\tau\in[0,2\tau_0]$.

Now, note that by \qref{e:G1} and the definition of $G$,
for $\tau\in[0,\tau_0)$ we have $\hat t+\tau<T$ and
\begin{equation}\label{e:G3}
G(w(\hat t+\tau;x(\hat t+\tau)),0,c(\hat t+\tau))=0.
\end{equation}
For $\tau\in[\tau_0,2\tau_0]$ we have $\hat t+\tau\in[T,T+\tau_0]$,
and we {\em define}
\begin{equation}\label{e:xc}
(x(\hat t+\tau),c(\hat t+\tau))=
(x(\hat t)+\hat\gamma(\tau), \hat c(\tau)).
\end{equation}
Then \qref{eq:xc0} and \qref{e:G1} hold for $0\le t\le T+\tau_0$, 
due to \qref{e:G2} for $\hat t+\tau\in[T,T+\tau_0]$.
We {\em claim} that \qref{e:xc} holds
for all $\tau\in[0,\tau_0)$ also, hence for all 
$\tau\in[0,2\tau_0]$.
From this claim it follows that
$x(t)$ and $c(t)$ are $C^1$ and \qref{eq:xc0} and \qref{e:G1} 
hold on $[0,T+\tau_0]$, so $T$ is not maximal.

To prove the claim, note that by \qref{e:wU0}--\qref{e:G3}
and the local uniqueness statement in Lemma~\ref{lem:decomp}
applied with $w=w(\hat t+\tau;x(\hat t))\in U_0$,
we have the following implication. For $\tau\in[0,\tau_0)$,
\begin{equation}
\mbox{if}\ \ z(\tau):=(x(\hat t+\tau)-x(\hat t),c(\hat t+\tau))\in
U_1,
\quad\mbox{then}\ \ z(\tau)=(\hat\gamma(\tau),\hat c(\tau))\in
\hat U_1.
\end{equation}
Since indeed $z(0)=(0,c(\hat t))\in U_1$ and 
$\hat U_1$ is a compact subset of $U_1$, however, we infer
by continuity that
$\sup\{ \tau\in[0,\tau_0): z(\tau)\in U_1\} =\tau_0$.
This proves the claim, and finishes the proof of the Proposition.
\end{proof}
\begin{remark}
We remark that this Lemma implies that the decomposition \qref{e:udecom}
can be continued as long as $|c(t)-c_0|$ and $\|v_2\|_{L^2_\ap}$ remain
sufficiently small, since $\|u_c-u_{c_0}\|_{L^2_\ap}\lesssim|c-c_0|$.  
\end{remark}

\section{Modulation equations and energy estimates}
\label{sec:mod-energy}

The decomposition described in \eqref{eq:decomp1}--\eqref{eq:orth}
yields a system of ordinary differential equations that govern the
modulating speed $c(t)$ and phase shift $x(t)$ of the main solitary wave.
In this section we describe these modulation equations,
and we provide estimates that control the 
energy norm of the combined perturbation
$v$ in terms of initial data and the modulation of the wave speed.
For $u=(q,r)$ we denote the energy density by 
\begin{equation}
\label{d.calE}
  \mathcal{E}(u)=\frac12\left( q^2+r^2+a(\pd_xq)^2+b(\pd_xr)^2\right).
\end{equation}

\subsection{Modulation equations}
Differentiate \eqref{eq:orth} with respect to $t$ and substitute
\eqref{eq:v2} into the resulting equation. Using the fact 
from \qref{e:zaeq} that
$\Lc^*\z_{2,c}^*=0$ and $\Lc^*\z_{1,c}^*=\z_{2,c}^*$ are both
orthogonal to $v_2$, it follows that for $i=1$ and $2$,
\begin{align*}
0 = &\ \frac{d}{dt}\la v_2(t), \z_{i,c(t)}^*\ra
-\la v_2,\mathcal{L}_{c(t)}^*\z_{i,c(t)}^*\ra
\\=&\ \dot{c}\la v_2,\pd_c\z_{i,c}^*\ra
+(\dot{x}-c)\la \D_yv_2,\z_{i,c}^*\ra
+\la l+\kone+\ktwo,\z_{i,c}^*\ra
. 
\end{align*}
Since $l= (\dot x-c)\z_{1,c}+\dot{c}\z_{2,c}$, by the biorthogonality
relations $\ip{\z_{i,c},\z_{j,c}^*}=\delta_{ij}$ 
from \qref{e:zaeq} we obtain that
$\dot x$ and $\dot c$ are determined by the {\em modulation equations}
\begin{equation}\label{eq:modeq0}
\begin{pmatrix}
1+\ip{\D_yv_2,\z_{1,c}^*} & \ip{v_2,\D_c\z_{1,c}^*} \\
\ip{\D_yv_2,\z_{2,c}^*} & 1+ \ip{v_2,\D_c\z_{2,c}^*}
\end{pmatrix}
\begin{pmatrix}
\dot x -c \\ \dot c
\end{pmatrix}
+
\begin{pmatrix}
\ip{\kone+\ktwo,\z_{1,c}^*} \\
\ip{\kone+\ktwo,\z_{2,c}^*}
\end{pmatrix}
=0.
\end{equation}

Our next lemma provides estimates for these modulation equations
in terms of the space $W_\vap$ with localized energy norm defined by 
\begin{gather}
\label{d.Wnorm}
\wnorm{v}
=\left(\int_\R e^{-2\vap|y|}\mathcal{E}(v(y))\,dy\right)^{1/2}.
\end{gather}

\begin{lemma}
  \label{lem:modulation}
Let $c_0>1$, $x_0\in\R$ and suppose $0<\nu\le \alpha<\frac12\alpha_{c_0}$.
Then there exist positive constants $\delta_2$ and $C$ 
with the following property.
Suppose the decomposition in Proposition~\ref{cor:decomp} 
holds on $[0,T]$ and suppose
$$
\sup_{t\in[0,T]}\left(|c(t)-c_0|+\|v(t)\|_{H^1}+\|v_1(t)\|_{H^1}
+\|v_2(t)\|_{H^1_\alpha}\right)\le \delta_2\,.$$
Then for $t\in[0,T]$,
\begin{align}
  \label{i:mod1}
 |\dot{x}(t)-c(t)|\le &\ C\vonew +
C \|v_2(t)\|_{H^1_\alpha}
(\|v_1(t)\|_{H^1}+\|v(t)\|_{H^1} +\|v_2(t)\|_{H^1_\alpha}) \,, 
\\
  \label{i:mod2}
 |\dot{c}(t)|\le &\  C\vonew 
+C \|v_2(t)\|_{H^1_\alpha} (\vonew +\|v_2(t)\|_{H^1_\alpha}) \,.
\end{align}
Furthermore, 
\begin{equation}
  \label{eq:modulation2}
\frac{d}{dt}\left(c(t)+\la \ti v_1(t), \z_{2,c(t)}^*\ra \right)
=O\left(\vonew ^2+\|v_2(t)\|_{H^1_\alpha}^2\right)\,.
\end{equation}
\end{lemma}
\begin{proof}
Note that for $\delta_2$ small enough, $2\ap<\ap_{c(t)}$ for all
$t\in[0,T]$.  And due to the estimates
\[
|\ip{\D_yv_2,\z_{j,c}^*}|
+|\ip{v_2,\D_c\z_{j,c}^*}| 
 \lesssim \|v_2\|_{H^1_\ap} \le\delta_2
\]
for $j=1,2$, the matrix in \qref{eq:modeq0}
is invertible with inverse $I+O(\|v_2\|_{H^1_\ap})$.

To estimate terms involving $k_1=f'(u_c)\ti v_1$, note 
$\zaone$ and $\zatwo$ are uniformly bounded in $L^2_{-\ap}$,
so $ |\ip{k_1,\z_{j,c}}| \lesssim \|k_1\|_{L^2_\ap}$ for $j=1,2$.
Since 
\[
f'(u_c)\ti v_1=-(B^{-1}\pd_x)
\begin{pmatrix}0 & 0\\ r_c & 2q_c \end{pmatrix}\ti v_1
-B^{-1}
\begin{pmatrix}0 & 0\\ r_c' & -q_c'
\end{pmatrix}\ti v_1,
\]
and $B\inv$ and $B\inv\D$ are bounded from $L^2_\ap$ to $H^1_\ap$,
we may deduce
\begin{equation}\label{i:k1}
\|k_1\|_{H^1_\ap} \lesssim 
\|q_c\ti v_1\|_{L^2_\ap} \lesssim
\|e^{-\apc |y|/2}\ti v_1\|_{L^2} \lesssim
\wnorm{\ti v_1},
\end{equation}
using $\ap<\half\apc$.  
(This estimate will be used also in section~\ref{s:stab}.)

Next we estimate terms involving $k_2$. Since $f$ is quadratic,
\[
k_2=f(v)-f(\ti v_1) = f'(\ti v_1)v_2+f(v_2)=f'(v)v_2-f(v_2).
\]
As for $k_1$, we find  
$ |\ip{k_2,\z_{1,c}}| \lesssim \|k_2\|_{L^2_\ap}$ and
\begin{equation}\label{i:k2}
\|k_2\|_{H^1_\ap} \lesssim 
\|v_2\|_{H^1_\alpha}
(\|v\|_{H^1}+\|v_2\|_{H^1})
\lesssim  \|v_2\|_{H^1_\alpha} (\|v\|_{H^1}+\|v_1\|_{H^1}) .
\end{equation}
Since $|\zatwo|\lesssim e^{-\apc|y|}\le e^{-2\ap|y|}$, however,
from Lemma \ref{cl:1} below we find the tighter estimate
\begin{align*}
|\la \ktwo,\z_{2,c}^*\ra|&\  \lesssim 
 \|e^{-2\ap|y|}(f'(\ti v_1)v_2+f(v_2))\|_{L^1}
\lesssim 
\wnorm{\ti v_1} \|v_2\|_{H^1_\alpha} +\|v_2\|_{H^1_\alpha}^2.
\end{align*}
Then directly we obtain \eqref{i:mod1} and \eqref{i:mod2}.

Next we prove \eqref{eq:modulation2}. Using \eqref{eq:v1}
and the fact that $\ip{\Lc\ti v_1,\zatwo}=\ip{\ti v_1,\Lc^*\zatwo}=0$,
we have
\begin{align*}
&\frac{d}{dt} \ip{\ti v_1,\z_{2,c(t)}^*}
 =
\ip{ \dot{x}(t)\pd_y\ti v_1+L\ti v_1 + f(\ti v_1),\zatwo }
+\dot{c}\,\ip{ \ti v_1,\pd_c\zatwo }
\\ & \quad = 
-\ip{f'(u_c)\ti v_1,\zatwo}
 +O\left( (|\dot{x}(t)-c(t)|+|\dot{c}(t)|)
\vonew +\vonew ^2\right)
\\ & \quad = -\ip{ \kone,\zatwo}
+O(\vonew ^2+\|v_2(t)\|_{H^1_\alpha}^2).
\end{align*}
Combining this with \qref{eq:modeq0} and \eqref{i:mod2}, we obtain
\eqref{eq:modulation2}. This completes the proof.
\end{proof}
For later use, we also note here that we have
\begin{equation}\label{i:vz2}
|\ip{\ti v_1(t),\z_{2,c(t)}^*}| \lesssim \vonew.
\end{equation}

\begin{lemma}
\label{cl:1}
  Let $|\alpha|< 1/\sqrt{b}$. Then for $p\in[1,\infty]$,
  \begin{gather}
\label{eq:cl11}
\|e^{-\alpha|x|}B^{-1} g\|_{L^p}
+\|e^{-\alpha|x|}B^{-1} \pd_xg\|_{L^p}\le C 
\|e^{-\alpha|x|}g\|_{L^p},
  \end{gather}
where $C$ is a positive constant depending only on $\alpha$.
\end{lemma}
\begin{proof}
Observe $B\inv g(x)=\intR 
\frac{1}{2\sqrt{b}}
e^{-|x-y|/\sqrt{b}}g(y)\,dy.$
Then for $j=0, 1$, 
$$
|e^{-\ap|x|}
(\dx^j B^{-1}g)(x)|\lesssim
\int_\R \kappa(x,y)
e^{-\ap|y|}
|g(y)|\,dy,
$$
where 
$\kappa(x,y)=
e^{-\ap|x|}e^{-|x-y|/\sqrt{b}}e^{\ap|y|}$.
Using the fact that 
$\sup_y\intR\kappa(x,y)\,dx + \sup_x\intR\kappa(x,y)\,dy<\infty$
we have \eqref{eq:cl11}.
\end{proof}

\subsection{Energy norm estimates on $v$}
We will estimate the energy norm of $v(t)$ by using the convexity of
the energy functional as was done for the case of FPU lattice 
models in \cite{FP2}.
Since the solitary wave
is not a critical point of the energy functional $E(u)$, the estimate
of $v(t)$ depends on the modulation of the speed $c(t)$.
\begin{lemma}
  \label{lem:sp-Ham}
Let $c_0>1$ and $u(0)=u_{c_0}(\cdot-x_0)+v_0$ for some $x_0\in\R$.
Let $\delta_3$ be a sufficiently small positive number and
$T\in[0,\infty]$.
Suppose that the decomposition of Proposition~\ref{cor:decomp}
exists for $t\in[0,T)$
and that 
$$\|v_0\|_{H^1}+\sup_{t\in[0,T)}(|c(t)-c_0|+\|v(t)\|_{H^1})\le \delta_3.
$$
Then
$$
\|v(t)\|_{H^1}^2\le C(\|v_0\|_{H^1}+|c(t)-c_0|)
\quad \text{for $t\in[0,T)$},$$
where $C$ is a positive constant depending only on $\delta_3$ and $c_0$.
\end{lemma}
\begin{proof}
Since the energy $E(u)$ is invariant under time evolution 
and spatial translation,
$$E(u(t))=E(u_{c_0}(\cdot-x_0)+v_0)=E(u_{c_0})+O(\|v_0\|_{H^1}).$$
Expanding $E(u(t))=E(u_{c(t)}+v(t))$ in a Taylor series about $u_{c(t)}$,
we have 
\begin{align*}
E(u(t)) =& E(u_{c(t)})
+\la E'(u_{c(t)}),v\ra+\frac12\la E''(u_{c(t)})v,v\ra
+O(\|v(t)\|_{H^1}^3).
\end{align*}
Since
$E'(u_c)=(Aq_c,Br_c)=:\zzone$ is a multiple of $\zatwo$ 
from Appendix~\ref{apx:Multzero}, by \eqref{eq:orth} we have
$$\la E'(u_{c(t)}),v(t)\ra=\la \ti v_1(t), \zzone\ra
=O(\|v_1(t)\|_{H^1})
\,.$$
Since $E''(u_c)=\operatorname{diag}(A,B)$ is positive definite,
there exist a positive constant $C'$ such that
\begin{align*}
\|v(t)\|_{H^1}^2 \le  C'(|E(u_{c(t)})-E(u_{c_0})|+\|v_0\|_{H^1}+\|v_1(t)\|_{H^1}
+\|v(t)\|_{H^1}^3)\,.
\end{align*}
If $\delta_3$ is sufficiently small, it follows
$$
\|v(t)\|_{H^1}^2 \le C(\|v_0\|_{H^1}+|c(t)-c_0|),$$
where $C$ is a positive constant depending only on $\delta_3$ and $c_0$.
Note that from \eqref{eq:energy} it follows 
$\|v_1(t)\|_{H^1}^2 \lesssim E(v_1(t))=O(\|v_0\|_{H^1}^2)$,
because $v_1(t)$ is a solution of \eqref{eq:v1}.
This completes the proof of Lemma \ref{lem:sp-Ham}.
\end{proof}

\section{Virial transport estimate}
\label{sec:virial}
In this section, we prove a virial lemma for small-energy 
solutions of \eqref{eq:v1} --- solutions of the `free' Benney-Luke system. 
This kind of result involves bounds on the transport of energy density,
measured using weighted integral quantities.
Essentially, this provides nonlinear estimates that
correspond to the fact that the solitary wave speed exceeds
the group velocity of linear waves in the present case 
($0<a<b$ and $c>1$) for the Benney-Luke equation \qref{e.bl}.

We start by observing that by a straightforward calculation,
we find that the energy density $\E(v_1)$ from \qref{d.calE}
satisfies a conservation law
\begin{equation} \label{e.EF}
\dt\E(v_1)=\dx\F(v_1),
\end{equation}
with the flux $\F = \F_2+\F_3$ where
\begin{align*}
\F_2(v_1) &= r_1 B\inv A q_1 + a (\dx q_1)(\dx r_1) ,\\
\F_3(v_1) &= 
 - r_1^2q_1 - br_1 \dx B\inv(r_1\dx q_1 + 2 q_1 \dx r_1).
\end{align*}

Let $\vap$ be a positive constant and $\tilde{x}(t)$ be a $C^1$-function.
We introduce a smoothed Heaviside function 
and a corresponding weighted energy by
\begin{equation}
\cha(x)=1+\tanh\vap x,
\qquad
\V(t)=\int_\R \cha(x-\tilde{x}(t))\E(v_1(t,x))\,dx.
\end{equation}
Note 
\[
\cha'(x)= \vap \psi_\vap(x)^2 \le 4\vap e^{-2\vap|x|}, \quad\text{ where 
\ \ $\psi_\vap(x)=\sech\vap x$.} 
\]
\begin{lemma} [Virial Lemma]
  \label{lem:virial} 
For any constant $c_1>1$, there exist positive numbers $\vap_0$,
$\delta_4$ and $\mu$ with the following property.
Given any $\vap\in(0,\vap_0)$, 
any $C^1$ function $\ti x(t)$ 
satisfying $\D_t\ti x(t)\ge c_1$ for all $t$, 
and any solution $v_1(t)$ to \qref{eq:v1} with
$\|v_0\|_{H^1}<\delta_4$, we have 
\begin{equation}
  \label{eq:virial}
\V(t)+\mu\vap\int_0^t\int_\R \psi_\vap(x-\tilde{x}(s))^2
\mathcal{E}(v_1(s,x))\,dx\,ds
\le \V(0).
\end{equation}
\end{lemma}
Lemma \ref{lem:virial} yields that $v_1(t)$ locally tends to $0$ as
 $t\to\infty$, with respect to any coordinate frame that moves at 
a speed strictly greater than one,
provided the energy is sufficiently small.
Before providing the proof, we establish two claims.
\begin{claim}
  \label{cl:virial2}
For every $u\in H^1(\R)$ we have 
$\int_\R (\E(u)+\F_2(u))\,dx \ge 0$.
\end{claim}
\begin{proof}
Due to Plancherel's identity, 
$\int_\R (\E(u)+\F_2(u))\,dx 
= \frac12\int_\R \hat u(\xi)^t D(\xi) \overline{\hat u(\xi)}\,d\xi$
where
\\ 
\begin{gather*}
 D(\xi)=\begin{pmatrix}
  1+a\xi^2 & S(\xi)^2+a\xi^2 \\ S(\xi)^2+a\xi^2 &   1+b\xi^2
\end{pmatrix}.
\end{gather*}
But by Gerschgorin's circle theorem, both eigenvalues 
$\kappa_1(\xi)$ and $\kappa_2(\xi)$ of $D(\xi)$ satisfy
\[
\kappa_j(\xi) > (1+a\xi^2)-(S(\xi)^2+a\xi^2) \ge 0,
\]
for $\xi\ne0$
since $0<a<b$. Thus $D(\xi)$ is positive definite and 
Claim \ref{cl:virial2} follows.
\end{proof}

\begin{claim} \label{cl:virial}
Let $0<\vap_0<1/\sqrt b$. Then for $\nu\in(0,\nu_0)$ we have
\[
\|[\psi_\vap,\dx]g\|_{L^2}+
\|[\psi_\vap,B^{-1}]g\|_{L^2}
=O(\vap\|\psi_\vap g\|_{L^2}).
\]
\end{claim}
\begin{proof}
The bound on the first commutator holds 
because $|\psi_\vap'/\psi_\vap|\le \vap$ uniformly. 
Note
\[
[\psi_\vap,B\inv]
= B\inv[B,\psi_\vap]B\inv
= -b B\inv (2\psi_\vap'\dx+\psi_\vap'') B\inv.
\]
Then since $|\psi_\vap''|\le 2\nu^2\psi_\vap$
and $ e^{-\vap|x|}\le \psi_\nu(x)\le 2e^{-\vap|x|}$, 
the bound on the second commutator follows 
by Lemma~\ref{cl:1}.
\end{proof}

\begin{proof}[Proof of Lemma \ref{lem:virial}]
Let $v_1=(q_1,r_1)$. 
By \qref{e.EF}, we compute
\begin{align}
\frac{d}{dt}
\int_\R \cha(x-\tilde x(t))
\E(v_1(t,x))\,dx  &= 
\int_\R \cha'(x-\tilde x(t))
(-(\D_t\tilde x)\E(v_1)- \F(v_1))\,dx 
\nonumber \\ & \le 
\int_\R \cha'(x-\tilde x(t))
(-c_1 \E(v_1)- \F(v_1))\,dx .
\end{align}
Now $\cha'(x-\ti x(t))=\vap\ti\psi^2$ where
$\ti\psi = \psi_\vap(x-\tilde{x}(t))$. 
Due to the commutator estimates of Claim \ref{cl:virial},
and then by Claim \ref{cl:virial2},
\begin{align}
& \int_\R \ti\psi^2 (-c_1\E(v_1)-\F_2(v_1)) \,dx  = 
\int_\R (-c_1\E(\ti\psi v_1)-\F_2(\ti\psi v_1))\,dx  
+ O(\vap\| \ti\psi v_1\|_{L^2}^2)
\nonumber \\
&\qquad \le 
(-c_1+1+O(\vap)) \int_\R \E(\ti\psi v_1) \,dx 
\le  
(-c_1+1+O(\vap)) \int_\R \ti\psi^2 \E(v_1) \,dx .
\end{align}
Moreover, 
$\int_\R \ti\psi^2 |r_1^2q_1|\,dx \le 
\|q_1\|_{L^\infty} \int_\R |\ti\psi r_1|^2\,dx$.
Writing $h(x)=r_1\dx q_1+2q_1\dx r_1$, we have the estimates
\begin{equation}
\|\ti\psi h \|_{L^2} \lesssim \|v_1\|_{L^\infty} \|\ti\psi v_1\|_{H^1},
\end{equation}
and by Lemma~\ref{cl:1},
\begin{align*}
\int_\R \ti\psi^2 |r_1\dx B\inv h |\,dx 
\le \|\ti\psi r_1\|_{L^2} \|\ti\psi \dx B\inv h \|_{L^2}
\lesssim
\|\ti\psi r_1\|_{L^2} \|\ti\psi h \|_{L^2}.
\end{align*}
It follows
\begin{equation}
  \label{eq:vir1}
\int_\R \ti\psi^2 |\F_3(v_1)| \,dx \lesssim 
\|v_1\|_{L^\infty} \int_\R \E(\ti\psi v_1)\,dx
\lesssim \|v_1\|_{L^\infty} \int_\R \ti\psi^2 \E(v_1)\,dx.
\end{equation}
By energy conservation \eqref{eq:energy} and the Sobolev imbedding theorem,
$\|v_1(t,\cdot)\|_{L^\infty} \lesssim E(v_0)^{1/2}$.
Thus, with $\mu=\frac12(c_1-1)$, say, 
if we choose $\vap_0$ and $\delta_4>0$ sufficiently small then 
it follows
\begin{equation}
  \label{eq:virial-diff}
\frac{d}{dt}\int_\R  \cha(x-\ti x(t)) \E(v_1(t,x))\,dx  +\mu
\int_\R \cha'(x-\ti x(t))\E(v_1(t,x))\,dx \le 0.
\end{equation}
Integrating this on $[0,t]$, we have \eqref{eq:virial}.
\end{proof}

\begin{corollary}
  \label{cor:virial}
Under the conditions of Lemma~\ref{lem:virial}, if 
there is a positive constant $\hat\sigma$ such that
$\D_t \ti x(t)\ge c_1+\hat\sigma$ for all $t$, then
\[
\intR \cha(x-\ti x(t))\E(v_1(t,x))\,dx
\le 
\intR \cha(x-\ti x(0)-\hat\sigma t)\E(v_0(x))\,dx,
\]
and this tends to 0 as $t\to\infty$.
\end{corollary}
\begin{proof}
Given any $t_1>0$ define $\ti x_1(t)=\ti x(t) - \hat\sigma(t-t_1)$.
Then $\ti x_1(t_1)=\ti x(t_1)$ and $\D_t\ti x_1(t)\ge c_1$ for all $t$. 
Using $\ti x_1$ in place of $\ti x$ in Lemma~\ref{lem:virial}, we find
\[
\intR \cha(x-\ti x(t_1))\E(v_1(t_1,x))\,dx
\le
\intR \cha(x-\ti x_1(0))\E(v_0(x))\,dx .
\]
\end{proof}


\section{Stability estimates with a priori smallness assumptions}
\label{s:stab}

For the remaining three sections, we fix $c_0>\sigma>c_1>1$.
Also fix $\ap\in(0,\frac12\ap_{c_0})$ and suppose that in $L^2_\ap$,
$\Lco$ has no nonzero eigenvalue satisfying $\Re\lambda\ge0$.
Let $\nu_0$ be given by the Virial Lemma,
and fix $\nu\in(0,\nu_0)$ with $\nu\le\ap$.
Define
\begin{align*}
&
\bmone(T)=\sup_{t\in[0,T]}\|v_1(t)\|_{H^1}+\|\ti v_1\|_{L^2(0,T;W_\vap)},
\\ &
\bmtwo(T)=\sup_{t\in[0,T]}\|v_2(t)\|_{H^1_\alpha}
+\|v_2\|_{L^2(0,T;H^1_\alpha)},
\\ & 
\bmv(T)=\sup_{0\le t\le T}\|v(t)\|_{H^1}^2,
\\ &
\bmc(T)=\sup_{t\in[0,T]}|c(t)-c_0|,\qquad
\bmx(T)=\sup_{t\in[0,T]}|\dot{x}(t)-c(t)|,
\\ & 
\bmtot(T) = \bmv(T)+\bmone(T)+\bmtwo(T) +\bmc(T)+\bmx(T).
\end{align*}
We shall first deduce
a priori bounds on $\bmc$, $\bmx$, $\bmv$ and $\bmone$ in terms of
$\|v_0\|_{H^1}$ and $\bmtwo$.
\begin{lemma}
\label{lem:apb1}
There exists a positive constant $\delta_5$ such that if
$\|v_0\|_{H^1}+\bmtot(T) \le \delta_5$, then
\begin{align}
\label{eq:bmone}
& \bmone(T) \lesssim \|v_0\|_{H^1},\\
\label{eq:bmv}
& \bmv(T) \lesssim \|v_0\|_{H^1}+\bmtwo(T)^2,\\
\label{eq:bmc}
& \bmc(T) \lesssim  \|v_0\|_{H^1}+\bmtwo(T)^2, \\
\label{eq:bm2}
& \bmx(T) \lesssim \|v_0\|_{H^1}+\bmtwo(T)^2.
  \end{align}
\end{lemma}
\begin{proof}
Energy conservation and the Virial Lemma imply \eqref{eq:bmone}. 
Lemma \ref{lem:sp-Ham} implies
\begin{equation}
  \label{eq:apb13}
  \bmv(T) \lesssim \|v_0\|_{H^1}+\bmc(T).
\end{equation}
Integrating \eqref{eq:modulation2} and using 
that $c(0)=c_0$, we find
\begin{equation}
  \label{eq:apb11}
  |c(t)-c_0|
 \lesssim  \bmone(T)+\bmone(T)^2+\bmtwo(T)^2.    
\end{equation}
Combining this with \eqref{eq:bmone} we obtain \eqref{eq:bmc}.
Then \qref{eq:bmv} follows from \qref{eq:bmc} and \qref{eq:apb13}.
Finally, by \eqref{i:mod1} we have 
\begin{align*}
  \label{eq:apb12}
\bmx(T)& \lesssim \bmone(T)+(\bmone(T)+\bmv(T)^{1/2}+\bmtwo(T))\bmtwo(T)
\\& \lesssim \bmone(T)+ \bmone(T)^2+\bmv(T)+ \bmtwo(T)^2,
\end{align*}
and combining this with \eqref{eq:bmone} and \eqref{eq:bmv} 
we obtain \eqref{eq:bm2}. 
\end{proof}

Now we will estimate $\bmtwo(T)$, making use of the recentering lemma
(Lemma~\ref{lem:lineardecay2}).
\begin{lemma}
  \label{lem:aprm3}
Let $\delta_5$ be as in Lemma \ref{lem:apb1}. If $\delta_5$ is sufficiently
small, then $\bmtwo(T)\lesssim \|v_0\|_{H^1}.$
\end{lemma}
\begin{proof}
We will prove Lemma \ref{lem:aprm3} by applying Lemma \ref{lem:lineardecay2}
to \eqref{eq:v2}.
By Lemma \ref{lem:modulation} and the definition of $l(s)$,
we have for $s\in[0,T]$, 
\begin{align*}
\|l(s)\|_{H^1_\alpha} \lesssim
&\  |\dot{x}(s)-c(s)|+|\dot{c}(s)|
\\ \lesssim &\  \vonews+
(\bmv(T)^{1/2}+\bmone(T)+\bmtwo(T))\|v_2(s)\|_{H^1_\alpha}.
\end{align*}
And by \qref{i:k1} and \qref{i:k2}, we have
\begin{align*}
\|\kone(s)\|_{H^1_\alpha} \lesssim \vonews, \qquad
\|\ktwo(s)\|_{H^1_\alpha} \lesssim  
(\bmv(T)^{1/2}+\bmone(T))\|v_2(s)\|_{H^1_\alpha}.
\end{align*}
Since $v_2(0)=0$, Lemma \ref{lem:lineardecay2} implies that there is
a constant $C_1$ such that 
\begin{align*}
\|v_2(t)\|_{H^1_\alpha}\le & C_1\int_0^t
e^{-\beta(t-s)/3}\left(\vonews 
+(\delta_5+\sqrt{\delta_5})
\|v_2(s)\|_{H^1_\alpha}\right)ds\,,
\end{align*}
for $t\in[0,T]$.
For $\delta_5$ small enough,
$C_1(\delta_5+\sqrt{\delta_5})\le \beta/12.$
Then by Gronwall's inequality,
\begin{equation}
\label{eq:ap5}
\|v_2(t)\|_{H^1_\alpha}\le C_1\int_0^t
e^{-\beta(t-s)/4}\vonews\,ds
\end{equation}
for $t\in[0,T]$.
Using Young's inequality and  $\bmone(T)\lesssim \|v_0\|_{H^1}$, we infer
$$
\bmtwo(T)=\sup_{t\in[0,T]}\|v_2(t)\|_{H^1_\alpha}+\|v_2\|_{L^2(0,T;H^1_\alpha)}
\lesssim  \|v_0\|_{H^1}.$$
This completes the proof of Lemma \ref{lem:aprm3}.
\end{proof}

\section{Proof of asymptotic stability}
Now we are in position to complete the proof of Theorem \ref{thm:1}
concerning the stability of solitary wave solutions.
\begin{proof}
Let $\delta_5$ be a positive constant given by Lemma \ref{lem:aprm3}.
Since $u(0)=\varphi_{c_0}(\cdot-x_0)+v_0$, $v_1(0)=v_0$,
it follows from 
Proposition~\ref{cor:decomp}
that if $\|v_0\|_{H^1}$ is small enough,
then there exists a $T>0$ such that
\begin{equation}
  \label{eq:bmall}
  \bmtot(T) \le \delta_5.
\end{equation}
Lemmas \ref{lem:apb1} and \ref{lem:aprm3} imply that
\begin{equation}
  \label{eq:bmall2}
\bmtot(T) \lesssim \|v_0\|_{H^1}\le \frac{\delta_5}{2} , 
\end{equation}
provided $\|v_0\|_{H^1}$ is sufficiently small.
Let $T_1\in(0,\infty]$ be the maximal time so that 
the decomposition in Proposition~\ref{cor:decomp}
persists for $t\in[0,T_1]$ and \eqref{eq:bmall} holds for any $T<T_1$. 
If $T_1<\infty$, then by \eqref{eq:bmall2} and Proposition~\ref{cor:decomp},
there exists $T_2>T_1$ such that 
the decomposition in Proposition~\ref{cor:decomp}
exists for $t\in[0,T_2]$ and \eqref{eq:bmall} holds for $T=T_2$,
which is a contradiction. Thus $T_1=\infty$ and
\eqref{eq:bmall} holds for $T=\infty$.
It follows
\[
\|u(t)-u_{c_0}(\cdot-x(t))\|_{H^1}=\|u_{c(t)}+v(t)-u_{c_0}\|_{H^1}
\lesssim \|v(t)\|_{H^1}+|c(t)-c_0|\lesssim \|v_0\|_{H^1}.
\]
Thus we obtain \eqref{OS}.

Next we will prove \qref{Phase2} and \qref{Phase1}.  
By Corollary \ref{cor:virial}, since 
$c_1<\sigma< \inf_t \dot x(t)$ for $\|v_0\|_{H^1}$ small enough, we have 
\begin{equation}
  \label{eq:ap4}
\vonew \lesssim \intR \cha(x-\sigma t-x_0)\E(v_1(t,x))\,dx
\to0
\quad\text{as $t\to\infty$.}
\end{equation}
Integrating \qref{eq:modulation2} and combining 
\qref{eq:ap4} with \qref{i:vz2} and the estimate
\begin{equation}
  \label{eq:L2t-bound}
\int_0^\infty
(\wnorm{\ti v_1(s)}^2+\|v_2(s)\|_{H^1_\ap}^2 )\,ds
\le \bmone(\infty)^2+\bmtwo(\infty)^2 
\lesssim \|v_0\|_{H^1}^2,  
\end{equation}
we conclude that $\cp =\lim_{t\to\infty}c(t)$ exists
and $|\cp -c_0|\lesssim \|v_0\|_{H^1}$, whence \qref{Phase2}.
Then, using \eqref{eq:ap4} with \eqref{eq:ap5} we find
\begin{equation}
  \label{eq:ap6}
\|v_2(t)\|_{H^1_\alpha}\to 0 \quad\text{as $t\to\infty$,}
\end{equation}
so by \qref{i:mod1} we obtain \qref{Phase1}.
For use below, we note that since $\vap=\theta\alpha$
where $\theta\in(0,1]$,
interpolating by H\"older's inequality we also have
\begin{equation}\label{i:v2nu}
\|v_2(t)\|_{H^1_\vap}\lesssim
\|v_2(t)\|_{H^1}^{1-\theta}\|v_2(t)\|_{H^1_\ap}^\theta
\to0 \quad\mbox{as $t\to\infty$}.
\end{equation}

It remains to prove \eqref{AS}.
For this we will 
use a {monotonicity argument} as in \cite{MM05}, 
applying virial estimates to $v(t,y)$. 
First, observe that since $\|u_{c(t)}-u_{\cp }\|_{H^1}\to0$,
it suffices to show 
\begin{equation}\label{ASnew}
\|u(t)-u_{c(t)}(\cdot-x(t))\|_{H^1(x\ge\sigma t)}
=\|v(t)\|_{H^1(y\ge \sigma t-x(t))} \to0
\qquad\mbox{as $t\to\infty$.}
\end{equation}

We will follow the arguments in the proof of the Virial
Lemma up to \qref{eq:vir1}.
Note that \eqref{eq:ap4} and \eqref{eq:ap6}
already imply that
\begin{equation}
  \label{eq:ap8}
\|v(t)\|_{H^1(y>0)}\lesssim
\int_\R\cha(y)\mathcal{E}(v(t,y))\,dy\to 0
\quad\mbox{as $t\to\infty$.}
\end{equation}
By \qref{eq:v}, $v$ satisfies
\begin{equation}\label{eq:vv}
\D_t v = L v+f(v) + \dot x(t) \D_y v+ l +f'(u_c)v,
\end{equation}
hence
\begin{equation}
\D_t\E(v)= \D_y\F(v) +\dot x\D_y \E(v)+\E'(v)(l+f'(u_c)v),
\end{equation}
where with $v=(\bar q,\bar r)$, $\ti v = (\ti q,\ti r)$ we write
\[
\E'(v)(\ti v)= \bar q \ti q + \bar r \ti r + 
a(\D_x\bar q)(\D_x \ti q) +b(\D_x\bar r)(\D_x \ti r).
\]
Next, for any $t_1>0$ given, let $\ti y(t) = c_1t -x(t) +
x(t_1)-c_1t_1$, and compute
\begin{gather}
\label{eq:tiy}
\frac{d}{dt}
\intR \cha(y-\tilde y(t)) \E(v(t,y))\,dy  = 
\int_\R \cha'(y-\tilde y(t))
(-c_1 \E(v)- \F(v) )\,dy  + I_1(t),
\\
I_1(t)  = \intR \cha(y-\ti y(t)) 
\E'(v)(l+f'(u_c)v)\,dy.
\nonumber
\end{gather}
As in the proof of Lemma~\ref{lem:virial}, 
writing $\cha'(y-\ti y(t))=\vap\ti\psi^2$,
for $\|v_0\|_{H^1}$ sufficiently small, 
we are guaranteed that
\begin{multline}
\intR \cha'(y-\tilde y(t))
(-c_1 \E(v)- \F(v) )\,dy 
\\ \le (-c_1+1+O(\vap)+O(\|v(t)\|_{H^1}))
\intR \vap\ti\psi^2 \E(v) 
\le 0 .
\end{multline}
Moreover, due to the localized nature of $l=(\dot x-c)\zone+\dot c\ztwo$ 
and $f'(u_c)v$, using Lemmas~\ref{lem:modulation} and \ref{cl:1} we have 
\begin{align*}
 \| \E'(v)(l)\|_{L^1} 
&\ \lesssim 
(|\dot x-c|+|\dot c|) \|v\|_{W_\ap}
\lesssim 
  \wnorm{\ti v_1}^2+\|v_2\|_{H^1_\ap}^2,
\\
 \| \E'(v)(f'(u_c)v)\|_{L^1}
&\ \le 
\|e^{-\ap|y|}\E'(v)\|_{L^2}
\|e^{\ap|y|}f'(u_c)v\|_{L^2}
\\
&\ \lesssim
\|v\|_{W_\ap}
\|e^{\ap|y|}q_c v\|_{L^2}
\lesssim
\|v\|_{W_\ap}^2
 \lesssim
 \wnorm{\ti v_1}^2+\|v_2\|_{H^1_\ap}^2,
\end{align*}
thus $\int_0^\infty I_1(t)\,dt\lesssim \|v_0\|_{H^1}^2$ 
by \qref{eq:L2t-bound}.
Integrating \qref{eq:tiy} for $t_1\le t$ we find
that since $\ti y(t_1)=0$,
\begin{equation}
\intR \cha(y-\ti y(t)) \E(v(t,y))\,dy  
\le
\intR \cha(y) \E(v(t_1,y))\,dy  
+ \int_{t_1}^\infty I_1(t)\,dt.
\end{equation}
The right hand side tends to zero as $t_1\to\infty$. 
Since $\sigma t-x(t)\ge\ti y(t)$ 
provided $ (\sigma-c_1)t \ge x(t_1)-c_1t_1 $, 
we can conclude that \qref{ASnew} holds.
This completes the proof of Theorem \ref{thm:1}.
\end{proof}

\section{Asymptotic stability in weighted spaces}
It remains to prove Theorem~\ref{thm:2}.
In addition to the assumptions of Theorem~\ref{thm:1}, assume \qref{eq:g}
where $\gwt:\R\to\R$ is increasing with $\gwt(x)=1$ for $x\le0$
and $\int_0^\infty\gwt(x)\inv\,dx<\infty$.

\subsection{Convergence of the phase shift}
\label{subsec:12.1}
In this subsection, we will prove \eqref{Phase3} by using
the Virial Lemma. 
\begin{proof}[Proof of \eqref{Phase3}]
With $3\hat\sigma=\sigma-c_1$ we have $\dot x(t)\ge c_1+3\hat\sigma$,
so by Corollary~\ref{cor:virial} we have
\[
\vonew^2\le
\intR \cha(x-x(t)) \E(v_1(t,x))\,dx
\le  \intR\cha(x-2\hat\sigma t) \E(v_0(x))\,dx
\]
for $t$ so large that $\hat\sigma t+x_0\ge0$.
Since $\cha(x)\le \min(2,2e^{2\vap x})$ for all $x$, and
$\gwt$ is increasing, we have
\[
\cha(x-2\hat\sigma t) \le
\begin{cases}
 2\gwt(x)^2/\gwt(\hat\sigma t)^2&\quad\text{if $x\ge \hat\sigma t$,}
\\ 2 e^{-2\vap\hat\sigma t}&\quad\text{if $x\le \hat\sigma t$.}
\end{cases}
\]
Therefore, 
we have 
\begin{equation}
  \label{eq:v1bound}
\intR \cha(y) \E(\ti v_1(t,y))\,dy
\le
2\intR\left(e^{-2\vap\hat\sigma t}+
\frac{\gwt(x)^2}{\gwt(\hat\sigma t)^2}\right)
\E(v_0(x))\,dx
\lesssim 
\Theta(t)^2,
\end{equation}
where
\[
\Theta(t) = 
(e^{-\vap\hat\sigma t}+\gwt(\hat\sigma t)\inv) 
(\|\gwt v_0\|_{L^2}+\|\gwt\D_xv_0\|_{L^2}).
\]
Due to \qref{eq:ap5} and Young's inequality, 
since $\int_0^\infty \Theta(t)\,dt<\infty$, 
\eqref{eq:v1bound} implies
\begin{equation}
\label{eq:v1v2L1}
\|\ti v_1\|_{L^1(0,\infty;W_\vap)}+\|v_2\|_{L^1(0,\infty;H^1_\alpha)}
<\infty.
\end{equation}
In view of \eqref{eq:v1v2L1} and the modulation estimates 
\qref{i:mod1}--\qref{i:mod2},
$\dot{x}(t)-c(t)$ and $\dot{c}(t)$ are integrable on $(0,\infty)$.
To show convergence of
\begin{equation}
  \label{eq:convformula}
\lim_{t\to\infty}(x(t)-\cp t)=x_0+
\int_0^t(\dot{x}(s)-c(s))\,ds+\int_0^t(c(s)-\cp )\,ds\,,
\end{equation}
it suffices to prove that 
$c(t)-\cp \in L^1(0,\infty)$.

By \qref{eq:ap5} and \qref{eq:v1bound} and the fact that $\Theta$ is decreasing,
we have
\begin{equation}\label{eq:v2bound}
\|v_2(t)\|_{H^1_\ap} \le
\left( \int_0^{t/2}+\int_{t/2}^t \right) e^{-\beta(t-s)/4} \Theta(s)\,ds
\le \frac4\beta\left( e^{-\beta t/8}\Theta(0)+ \Theta(t/2)\right).
\end{equation}
Then for large $t\ge0$, since $\int_0^\infty \Theta(t)\,dt<\infty$,
\begin{equation}
  \label{eq:intbound}
\int_t^\infty
\left(\vonews^2+\|v_2(s)\|_{H^1_\alpha}^2\right)\,ds
\lesssim 
e^{-\beta t/8}+\Theta(t/2) .
\end{equation}
Since $\vonew$ is integrable,
it follows by integrating \eqref{eq:modulation2} and using \qref{i:vz2} that
\begin{equation}
  \label{eq:convformula2}
c(t)-\cp =\int_\infty^t\dot{c}(s)\,ds
= 
O\left(
\vonew+
\int_t^\infty
  \left( \vonews^2
+\|v_2(s)\|_{H^1_\alpha}^2\right)\,ds\right)  
\end{equation}
is integrable on $(0,\infty)$.  
Now letting 
\begin{equation}
  \label{eq:convformula3}
\xp:=x_0+\int_0^\infty(\dot{x}(s)-c(s))\,ds+\int_0^\infty(c(s)-\cp )\,ds,
\end{equation}
we obtain \eqref{Phase3}.
\end{proof}

\subsection{Exponentially localized data}
\label{subsec:12.2}
In this subsection, we will prove 
\eqref{Exp-phase-convergence}--\eqref{Exp-convergence}.
To begin with, we will prove exponential decay of $v_1(t)$.
\begin{lemma}
\label{lem:v1-exponential}
There is a positive constants $\hat C$ such that 
if $\|v_0\|_{H^1}\le\delta_4$ 
and $v_0\in H^1_\apo$ for some $\apo\in(0,\vap_0)$,
then for all $t\ge0$,
\begin{equation}
  \label{eq:v1-exponential}
\left\| e^{\apo(x-c_1 t)}v_1(t,\cdot)\right\|_{H^1}
\le \hat C \|v_0\|_{H^1_\apo}.
\end{equation}
\end{lemma}
\begin{proof}
Observe that $\bar{\chi}^n(t,x):=e^{2\apo n}\chi_{\apo}(x-c_1t-n)\to
e^{2\apo(x-c_1t)}$ monotonically as $n\to\infty$.
Then Lemma~\ref{lem:virial} implies that
for every $n\in\N$,
\begin{equation*}
\int_\R\bar{\chi}^n(t,x)\mathcal{E}(v_1(t,x))dx
\le \int_\R \bar{\chi}^n(0,x)\mathcal{E}(v_0(x))dx\,.
\end{equation*}
Letting $n\to\infty$, by using Beppo Levi's theorem we obtain
\begin{equation}
\label{eq:v1-exp2}
\int_\R e^{2\apo(x-c_1t)}\mathcal{E}(v_1(t,x))dx \le
\int_\R e^{2\apo x}\mathcal{E}(v_0(x))dx\,.
\end{equation}
Eq. \eqref{eq:v1-exponential} immediately follows.
\end{proof}

\begin{proof}[Proof of \eqref{Exp-phase-convergence}--\eqref{Exp-convergence}]
We suppose $v_0\in H^1_{\apo}$ where $0<\apo<\min(\nu_0,\ap)$.
Let $\gamma_1=\apo(\sigma-c_1)$.
Lemma \ref{lem:v1-exponential} implies that
since $\dot x(t)>\sigma$,
\begin{equation}
  \label{eq:v1w-exp}
\vonew \lesssim 
\|\ti v_1(t)\|_{H^1_\apo}\lesssim
e^{-\apo(x(t)-c_1t)}\lesssim
e^{-\gamma_1 t}.
\end{equation}
By \eqref{eq:v1w-exp}  and \eqref{eq:ap5}, we have 
\begin{equation}
  \label{eq:v2w-exp}
\|v_2(t)\|_{H^1_\alpha}\lesssim e^{-\gamma_2 t}
\end{equation}
for  $\gamma_2$ satisfying $0<\gamma_2< \min(\gamma_1, \beta/4)$.
Interpolating as we did in \qref{i:v2nu} then yields 
\[
\|v_2(t)\|_{H^1_\apo}\lesssim e^{-\gamma t},
\qquad \gamma=\gamma_2\ap_1/\ap.
\]
Thus by Lemma \ref{lem:modulation},
\begin{equation}
  \label{eq:Exp-phase-convergence}
|c(t)-\cp |=O(e^{-\gamma t})\quad\text{and}\quad
|x(t)-\cp t-\xp|=O(e^{-\gamma t})\quad\text{as $t\to\infty$.}
\end{equation}
Let $\xot=x(t)-\cp t-\xp$. 
Combining \eqref{eq:v1w-exp} and \eqref{eq:v2w-exp} with
\eqref{eq:Exp-phase-convergence}, we obtain
\begin{align}
 \|u_\cp- u(t,& \cdot+\,\cp t+\xp)\|_{H^1_\apo} 
= e^{\apo\xot}
 \| u_\cp(\cdot+\xot)-u(t,\cdot+x(t))\|_{H^1_\apo}
\nonumber\\ &\quad = 
e^{\apo\xot}
\|u_\cp(\cdot+\xot)-u_{c(t)} -\ti v_1(t) -v_2(t)\|_{H^1_\apo}
\nonumber\\ &\quad \lesssim 
|\xot|+|c(t)-\cp|+ \|\ti v_1(t)\|_{H^1_\apo}+ \| v_2(t)\|_{H^1_\apo}
= O(e^{-\gamma t}).
\label{eq:uucp}
\end{align}
Thus we prove \eqref{Exp-phase-convergence} and \eqref{Exp-convergence}.
\end{proof}

\subsection{Polynomially localized data}
In this subsection, we will prove \eqref{Poly-convergence},
essentially as an immediate consequence of the arguments of
subsection \ref{subsec:12.1}.
Below, let $x_+=\max(0,x)$, and let $\rp>1$ be constant.

\par
First, we remark that for a localized perturbation 
with $\gwt v_0\in H^1$ for $\gwt(x)=(1+x_+)^{\rp}$,
Eqs. \eqref{eq:v1bound} and \eqref{eq:v2bound} imply
\begin{equation}
  \label{eq:poly-decay}
\|\chi_{\vap/2}\tilde{v}_1(t)\|_{H^1}
+\|v_2(t)\|_{H^1_\ap}\lesssim (1+t)^{-\rp}.
\end{equation}
By \eqref{eq:convformula2} and \eqref{eq:poly-decay}, 
and the fact that $ \vonew
\lesssim \|\chi_{\vap/2}\tilde{v}_1(t)\|_{H^1}$,
\begin{equation}
\label{eq:c-polyconv}
\begin{split}
|c(t)-\cp | \lesssim &
(1+t)^{-\rp}+\int_t^\infty(1+s)^{-2\rp}\,ds\lesssim (1+t)^{-\rp}
\end{split}
\end{equation}
provided $\rp>1$.
Since $x(0)=x_0$ and
$\dot{x}(t)-c(t)=O(\vonew +\|v_2(t)\|_{H^1_\alpha})$
by Lemma \ref{lem:modulation}, 
\begin{equation}
\label{eq:x-polyconv}
\begin{split}
x(t)-\cp t-\xp=& \int_\infty^t
(\dot{x}(s)-c(s))\,ds+\int_\infty^t(c(s)-\cp)\,ds
= O\left((1+t)^{-\rp+1}\right)
\end{split}
\end{equation}
follows from \eqref{eq:poly-decay}, \eqref{eq:c-polyconv},
and \eqref{eq:convformula3}.
Now \eqref{Poly-convergence} follows from \eqref{eq:poly-decay},
\eqref{eq:c-polyconv} and \eqref{eq:x-polyconv} 
in a manner very similar to \qref{eq:uucp}, using the fact
that 
$ \chi_{\nu/2}(x)\lesssim \chi_{\ap/2}(x) \sim \min(1,e^{\ap x}) $.

%% file: mpq-ham.tex
\section{Hamiltonian and variational structure}
\label{ap-ham}
The Benney-Luke equation \qref{e.bl} has a Hamiltonian structure which
we now describe. Also we show that the solitary-wave profile is an 
infinitely indefinite critical point of the naturally associated 
energy-momentum functional.  

We modify slightly the form in \cite{PQ99} to use $q$
in place of $\bphi$.  In terms of the conjugate momentum variable
\begin{equation}\label{e.pdef}
p= r+B^{-1}\left( \frac12 q^2\right),
\end{equation}
the Hamiltonian is given by
\begin{align} 
\CH &= \frac12 \int_\R  rBr+qAq\,dx 
\nonumber \\
&= \frac12 \int_{\R}
\left(p-B^{-1}\left(\frac12 q^2\right) \right) B\left(p- B^{-1}\left(\frac12
q^2\right)\right)+ qAq\,dx .
\label{e.hamdef}
\end{align}
We find that formally, taking variations with respect to $(q,p)$,
\[ \delta\CH  = 
\begin{pmatrix} 
-rq + Aq \\[5pt]
Br 
\end{pmatrix},
\]
and that \qref{e.qr} is equivalent to the following system
in Hamiltonian form,
\begin{equation} \label{ham}
\dt\begin{pmatrix} q \\ p \end{pmatrix} =
\CJ\, \delta\CH ,
\qquad \CJ = 
\begin{pmatrix} 0 & \partial_xB^{-1} \\ \partial_x B^{-1} & 0 \end{pmatrix}.
\end{equation}
Due to the translation invariance of the Hamiltonian,
Noether's Theorem assures the existence of the conserved
momentum functional
\[
\CN= \int_{\R}q Bp\,dx,
\]
with the property that $\CJ \delta\CN = \D_x$.
Solitary waves with speed $c> 0$ have profiles given as the stationary points
of the energy-momentum functional
\begin{equation}
\CH_c = \CH + c \CN.
\end{equation}
Noting that 
\[
\delta \CH_c = 
\begin{pmatrix}
-rq + Aq +cBp \\[5pt]
Br +c Bq
\end{pmatrix}
\]
one checks that solutions of $\delta \CH_c = 0$ satisfy
\qref{e.sol0}, 
and \qref{e.sol}-\qref{e.sol2} yields the localized solutions
for $c^2 >1$. 

The classic variational approach to proving orbital
stability for solitary waves \cite{Ben72,GSS1}
is based on showing that the second variation $\delta^2\CH_c$ 
has definite sign when subject to a finite number of constraints
induced by time-conserved quantities.
Here, at a critical point $(q,p)=(q_c,p_c)$, 
in terms of a variation $(\dot q,\dot r)=(\dot q,\dot p - B\inv(q\dot q))$ 
we can express the second variation as
\begin{align}
\ipbig{ 
\begin{pmatrix}\dot q \\ \dot r  \end{pmatrix},
\delta^2\CH_c 
\begin{pmatrix}\dot q \\ \dot r \end{pmatrix}
} &= 
\int_\R 
\begin{pmatrix}\dot q \\ \dot r  \end{pmatrix} ^T
\begin{pmatrix} A+3cq & cB \\ cB & B \end{pmatrix}
\begin{pmatrix}\dot q \\ \dot r \end{pmatrix}\, dx
\nonumber
\\ &=
\int_\R \dot q (A-c^2B+3cq)\dot q +
(\dot r+c\dot q) B (\dot r+c\dot q) \,dx .
\end{align}
The operator $B=I-b\dx^2$ is positive. However, since
$c^2 >1$ and $b>a$, the operator
\[
L_c=A-c^2B+3cq = (1-c^2)+(bc^2-a)\dx^2 + 3cq
\]
has the interval $(-\infty,1-c^2]$ as continuous spectrum. 
Zero is an eigenvalue, with eigenfunction $\dx q$ due to \qref{e.sol}.
Since this eigenfunction changes sign exactly once, oscillation theory implies
that $L_c$ has exactly one positive eigenvalue. 
Thus, $L_c$ is strictly negative except for two directions which are
associated with the two degrees of freedom of the solitary wave, while
$B$ is a positive operator. It follows that $\delta^2 \CH_c$ is infinitely
indefinite.  
This situation also occurs in the full water wave equations 
\cite{BS89} and in other Boussinesq-type nonlinear wave equations
having two-way wave propagation \cite{Sm92,PSW95}.

%% file: mpq-appx1.tex

\section{Multiplicity of the zero eigenvalue}
\label{apx:Multzero}

Here our aim is to prove part (iv) of Lemma~\ref{l.lin1},
and determine the generalized kernels of both $\Lc$ and $\Lc^*$. 
We will show that $\lambda=0$ is an eigenvalue of the operator $\Lc$
with algebraic multiplicity two and geometric multiplicity one,
in the space $L^2_\ap$, $0<\ap<\apc$.

Let us write $(q,r)$ for $(q_c,r_c)$ below for simplicity.
By differentiating the solitary wave equations \qref{e.sol0}
with respect to $x$ and $c$ it follows directly that the functions
\begin{equation} 
\zone = \begin{pmatrix}\D_x q\cr \D_x r\end{pmatrix}, 
\qquad
\ztwo = -\begin{pmatrix}\D_c q\cr \D_c r\end{pmatrix},
\end{equation}
satisfy $\Lc\zone=0$, $\Lc\ztwo=\zone$.

By basic asymptotic theory 
for ODEs (after multiplying the second component by $B$),
solutions of $\Lc z=0$ satisfy $z(x)\sim v e^{\mu x}$
as $x\to\infty$, where $v\in\R^2$ and
where $\mu$ is an eigenvalue of the characteristic matrix,
satisfying
\[
\det \begin{pmatrix} c\mu & \mu\cr (1-a\mu^2)\mu &c\mu(1-b\mu^2)
\end{pmatrix}
=\mu^2((c^2-1)-(bc^2-a)\mu^2)=0.
\]
The roots are $\mu=\pm\apc $ and the double root $\mu=0$,
so any solution of $\Lc z=0$ that lies in 
the space $L^2_\ap$ decays exponentially to zero as $x\to\infty$
and must be a constant multiple of $\zone$. Thus $\lambda=0$
has geometric multiplicity one.

Next we treat the adjoint $\Lc^*$.
In the following lemma, $\D\inv$ denotes a
right inverse for $\D$ on the space $L^2_{-\ap}$ dual to $L^2_\ap$,
defined by $\D\inv g(x)=\int_{-\infty}^x g(y)\,dy$.
\begin{lemma} \label{l:adj}
Suppose $0<\ap<\apc$, and let
\begin{equation}\label{d:z2s}
\zzone =
\begin{pmatrix}Aq\cr Br\end{pmatrix}, \qquad
\zztwo = -c\begin{pmatrix}
q (\D\inv  \D_cq)+ B\D\inv \D_c p 
\\ B\D\inv \D_c q
\end{pmatrix} ,
\end{equation}
where $p=r+B\inv (\frac12 q^2)$.
Then $\zzone$ and $\zztwo$ lie in $H^1_{-\ap}$ 
and satisfy
$\Lc^*\zzone=0$,  $\Lc^*\zztwo=\zzone$.
Moreover, 
\begin{equation}\label{eq:duals}
\begin{pmatrix}
\ip{\zone,\zzone}  & \ip{\ztwo,\zzone} \\ 
\ip{\zone,\zztwo}  & \ip{\ztwo,\zztwo}  
\end{pmatrix}
=
\begin{pmatrix}
0 & -\beta_0 \\
-\beta_0 & \beta_1 
\end{pmatrix},
\end{equation}
with
\begin{equation} \label{d:beta01}
\beta_0 = \frac{d}{dc}E(u_c) >0,
\qquad \beta_1 
= c \left(\frac{d}{dc}\intR  q \right)\left(\frac{d}{dc}\intR p\right). 
\end{equation}
\end{lemma}
For use in Part II, we define vectors biorthogonal to $\zone$, $\ztwo$
via 
\begin{equation}\label{d:zz}
\begin{pmatrix}
\zaone\\ \zatwo
\end{pmatrix}
=
\begin{pmatrix}
\theta_1 & \theta_0\\
\theta_0 & 0
\end{pmatrix}
\begin{pmatrix}
\zzone \\ \zztwo
\end{pmatrix},
\qquad
\begin{pmatrix}
\theta_0 \\ \theta_1
\end{pmatrix}
=
-\frac{1}{\beta_0^2}
\begin{pmatrix}
\beta_0 \\ \beta_1
\end{pmatrix}.
\end{equation}
Then for $i,j=1,2$,
\begin{equation}\label{e:zaeq}
\ip{\z_{i,c},\z_{j,c}^*}=\delta_{ij},
\quad 
\Lc^* \zaone = \zatwo,
\quad 
\Lc^* \zatwo=0. 
\end{equation}

Taking the lemma for granted temporarily, 
we claim it follows that $\lambda=0$ has algebraic multiplicity 
exactly equal to two. Suppose the multiplicity is higher.
Then there exists $\z_3\in L^2_\ap$ with $\Lc\z_3=\ztwo$.
But then, by Lemma~\ref{l:adj}, since $\zzone\in H^1_{-\ap}$,
\[
0=\ip{\z_3,\Lc^*\zzone} = \ip{\Lc\z_3,\zzone}=\ip{\ztwo,\zzone}=-\beta_0\ne0.
\]
This contradiction shows $\z_3$ cannot exist, 
hence the multiplicity is exactly two.

\begin{proof}[Proof of Lemma~\ref{l:adj}]
It is straightforward to check that 
\begin{equation}
\Lc^* \zzone = 
\begin{pmatrix}\label{d:z1s}
-c\D &  (-A\D -c \D q+2cq')B\inv \cr 
-\D &
(-cB\D+2\D q- q')B\inv
\end{pmatrix}
\begin{pmatrix}Aq\cr -cBq\end{pmatrix}
=0.
\end{equation}
Next, we compute that
\begin{equation}
\ip{\zone,\zzone}=\intR((\D_x q)Aq+(\D_x r)Br )\,dx = 0,
\end{equation}
\[
\ip{\ztwo,\zzone}=-\intR((\D_c q)Aq+(\D_c r)Br )\,dx 
= -\frac{dE}{dc},
\]
where
\begin{equation}\label{d.H1}
E =E(u_c) = 
\frac12 \intR (qAq+rBr)\,dx = \frac12 \intR 
\left((1+c^2)q^2 + (a+bc^2)(\D_xq)^2\right)\,dx. 
\end{equation}
Using the facts that 
\[
\intR \sech^4\frac x2\,dx = \frac83, 
\qquad \intR\sech^4\frac x2\tanh^2\frac x2\,dx 
= \frac8{15},
\]
from the explicit expression \qref{e.sol2} for $q$, we find that  
\begin{eqnarray}
E &=& 
\frac43 (1+c^2)\frac{(c^2-1)^2}{c^2} \apc \inv
+ \frac4{15} (a+bc^2)\frac{(c^2-1)^2}{c^2} \apc \\
&=& \frac{4\rho^2}{15(\rho+1)} \left(
5(\rho+2)\sqrt{b+\frac{b-a}\rho}+(b\rho+(b+a))\sqrt{\frac{\rho}{b\rho+b-a}}
\right),
\nonumber
\end{eqnarray}
where $\rho=c^2-1$.
From this expression it is evident that $dE/dc>0$ for $c>1$. 

Next we find some $\zztwo\in H^1_{-\ap}$
such that $\Lc^*\zztwo=\zzone$.  Writing $\zztwo=(\ti q,\ti r)$,
this means
\begin{align*}
\Lc^*\zztwo =
\begin{pmatrix}
-c\D\ti q + (-A\D -cq\D+cq')B\inv \ti r \\
-\D\ti q + (-cB\D +2q\D+q')B\inv \ti r 
\end{pmatrix} = 
\begin{pmatrix} Aq\\ -cBq\end{pmatrix}
\end{align*}
Eliminating $\ti q$ and comparing with the equation obtained by
differentiating \qref{e.sol} in $c$, 
\begin{equation}
(A-c^2B+3c q)\D_c q = 2cBq-\frac32 q^2= \frac1c (Aq+c^2Bq),
\end{equation}
we may choose $\ti r=-cB\D\inv\D_cq$. 
Then since $2q\D+q'=\D q+q\D$,
\begin{align*}
\D(\ti q -q B\inv\ti r)=  (-cB+q)\D B\inv \ti r +c Bq= 
cB(q+c\D_cq)-cq\D_cq =-cB\D_cp.
\end{align*}
Hence $\zztwo$ is given by \qref{d:z2s}, and $\zztwo\in H^1_{-\ap}$.
Moreover, we find
\begin{equation}
\ip{\zone,\zztwo}= \ip{\Lc \ztwo,\zztwo} = \ip{\ztwo,\Lc^*\zztwo}
=\ip{\ztwo,\zzone}=-\frac{dE}{dc}=-\beta_0.
\end{equation}

Finally, we compute $\ip{\ztwo,\zztwo}$.  
Since $\D_cp=\D_cr+B\inv(q\D_cq)$, 
we can write
\[
\pmat{\D_cq \\ \D_cp} 
= \CT 
\pmat{\D_cq \\ \D_cr}, 
\quad
\CT =\pmat{ I & 0\\ B\inv q & I},
\quad \CT^{-*} =\pmat{ I & -qB\inv\\ 0 & I},
\]
and we note
\[
\zztwo = 
-c \pmat{ I & qB\inv\\ 0 & I}
B\D\inv \pmat{\D_cp \\ \D_cq} =
-c \CT^* \CJ\inv \pmat{ \D_cq\\\D_cp},
\]
with $\CJ$ as in \qref{ham}.  Hence we find
\begin{align}
&\ip{\ztwo,\zztwo} = \ip{\CT\ztwo,\CT^{-*}\zztwo} = 
c\ipbig{
\pmat{\D_cq \\ \D_cp} ,\CJ\inv
\pmat{\D_cq \\ \D_cp} 
} \nonumber \\
&\qquad = c \intR 
(\D_cq) B\D\inv (\D_cp) 
+(\D_cp) B\D\inv(\D_cq ) 
= c \left(\intR \D_c q \right)\left(\intR \D_c p\right). 
\end{align}
\end{proof}

\section{Null multiplicity of the zero characteristic value}
\label{apxNullmult}

In order to apply the Gohberg-Sigal theory, we need to show that 
(i) for the bundle $\ww(\lambda)$, $\lambda=0$ is a 
{\sl characteristic value of null multiplicity at least two},
and 
(ii) for the KdV bundle $\WW_0(\Lambda)$, $\Lambda=0$
is the only characteristic value satisfying $\Re\Lambda> -\hb$,
and has null multiplicity {\sl no more than two}.
According to what this means in the terminology of \cite{GS},
we need to prove the following.
\begin{lemma}\label{l.nullmult}
For some nontrivial analytic map $\lambda\mapsto \psi(\lambda)\in L^2_\ap$, 
$\|\ww(\lambda)\psi(\lambda)\|_\ap = o(\lambda)$
{ as $|\lambda|\to0$}.
\end{lemma}

\begin{lemma}\label{l.nullkdv}
Suppose $\hap\in(0,(b-a)^{-1/2})$ and $\hb=\hap(1-(b-a)\hap^2)$.
Then 
$\WW_0(\Lambda)$ is invertible in $L^2_\hap$
whenever $\Re\Lambda>-\hb$ and $\Lambda\ne0$.
Moreover, $\WW_0(0)$ has one-dimensional kernel, 
and for no nontrivial analytic map $\Lambda\mapsto\psi(\Lambda)$ 
do we have
$\|\WW_0(\Lambda)\psi(\Lambda)\|_\ap = o(\Lambda^2)$
as $|\Lambda|\to0$.
\end{lemma}

For the proof of Lemma~\ref{l.nullkdv} see the proof of
Proposition 12.3 in Appendix C of \cite{PS2}. 
(The KdV bundle $W_0(\lambda)$ there differs from $\WW_0(\Lambda)$
here by a simple scaling.)

To prove Lemma~\ref{l.nullmult} is a simple calculation when done in the right way.
The trick is to apply the transformation in \qref{e.qr2}
to the original solitary-wave equations \qref{e.qr}, {\em then} differentiate
with respect to $x$ and $c$.
Using \qref{e.sim} with $\lambda=0$ we find that 
\begin{eqnarray}
-\Q_+(Sq+r) + B\inv(rq'+2qr') = 0,\\
-\Q_-(-Sq+r)+ B\inv(rq'+2qr') = 0.
\end{eqnarray}
Thus,  with 
\[
\rho:= 
-\Q_+(Sq+r) =
-\Q_-(-Sq+r)
\]
we have
\begin{equation}
r = \frac12(-\Q_+\inv-\Q_-\inv)\rho,
\qquad q = \frac12 S\inv(-\Q_+\inv+\Q_-\inv)\rho,
\end{equation}
\begin{equation}\label{e.rho}
\rho+ B\inv(rq'+2qr') = 0.
\end{equation}
Differentiating these equations with respect to $x$, we find
\begin{equation}
\D_xr = \frac12(-\Q_+\inv-\Q_-\inv)\D_x\rho,
\qquad \D_x q = \frac12 S\inv(-\Q_+\inv+\Q_-\inv)\D_x\rho,
\end{equation}
\begin{equation}
\D_x\rho+B\inv( q'+2q\D_x)\D_xr + B\inv(r\D_x+2r')\D_xq =0.
\end{equation}
This yields
\begin{equation}\label{e.nul0}
(I+(R_q+R_r)(-\Q_+\inv)+(R_q-R_r)(-\Q_-\inv) )\D_x\rho = 
\ww(0)\D_x\rho=0.
\end{equation}
Next we differentiate with respect to $c$.
Since $\Q_\pm = c\D_x\pm S\D_x$ we have
\[
\D_c(Q_\pm\inv\rho) = Q_\pm\inv\D_c\rho-Q_\pm^{-2}\D_x\rho.
\]
Hence 
\begin{eqnarray}
\D_cr &=& \frac12(-\Q_+\inv-\Q_-\inv)\D_c\rho
+\frac12(\Q_+^{-2}+\Q_-^{-2})\D_x\rho, \\
\qquad \D_c q &=& \frac12 S\inv(-\Q_+\inv+\Q_-\inv)\D_c\rho
+\frac12 S\inv(\Q_+^{-2}-\Q_-^{-2})\D_x\rho. 
\end{eqnarray}
and therefore
\begin{eqnarray}
0&=& (I+(R_q+R_r)(-\Q_+\inv)+(R_q-R_r)(-\Q_-\inv)) \D_c\rho
\\&& \nonumber \ \qquad +\ ((R_q+R_r)\Q_+^{-2} + (R_q-R_r)\Q_-^{-2})\D_x\rho. 
\end{eqnarray}
Since 
\[
W'(\lambda)= -(R_q+R_r)(\lambda-\Q_+)^{-2}-(R_q-R_r)(\lambda-\Q_-)^{-2}, 
\]
this means 
\begin{equation}\label{e.nul1}
\ww(0)\D_x\rho-W'(0)\D_c\rho=0.
\end{equation}
Since $\D_x\rho$ and $\D_c\rho$ belong to the weighted space $L^2_\ap$,
combining \qref{e.nul0} and \qref{e.nul1} we obtain
\[
\ww(\lambda)(\D_x\rho-\lambda\D_c\rho) = o(\lambda)
\qquad\mbox{as $|\lambda|\to0$},
\]
in $L^2_\ap$, and this finishes the proof of Lemma~\ref{l.nullmult}.

%% file: mpq-appx2.tex


\section{Exponential linear stability via recentering}
\label{sec:appendix-L}
To begin, we extend the linear stability estimate from
Theorem~\ref{t.splin} to
Sobolev spaces $H^n(\R)$ spaces of arbitrary order.
\begin{proposition}
\label{prop:Hnlinear} 
Fix $c>1$ and $\alpha$ with $0<\ap<\apc$, and let $n\ge0$ be an integer. 
Assume that $\Lc$ has no nonzero eigenvalue $\lambda$
satisfying $\Re\lambda\ge 0$. 
Then there exist positive constants $K_n$ and $\beta$ such that
for all $t\ge0$, 
\begin{equation}\label{i.semi2}
\|e^{\Lc t} Q_cz\|_{H^n_\alpha} \le K_ne^{-\beta t}\|z\|_{H^n_\alpha}\,,
\end{equation}
where $Q_c=I-P_c$ is the spectral projection 
complementary to the generalized kernel of $\mathcal{L}_c$.
\end{proposition}
\begin{proof}
Since $1$ is in the resolvent set of $\mathcal{L}_c$ by Lemma \ref{l.lin1},
and $Q_c$ commutes with $\Lc$, we see that 
$(1-\mathcal{L}_c)^n$ is an isomorphism from $Q_cH^n_\alpha$ to
$Q_cL^2_\alpha$. 
Applying Theorem \ref{t.splin}, we find
\begin{align*}
\|e^{\Lc t} Q_cz\|_{H^n_\alpha}
\lesssim \|(1-\mathcal{L}_c)^ne^{\Lc t} Q_cz\|_{L^2_\alpha}\lesssim
e^{-\beta t}\|z\|_{H^n_\alpha}.
\end{align*}
This completes the proof.  
\end{proof}

Our main goal in this appendix is to prove Lemma \ref{lem:lineardecay2} 
by a recentering argument. Such arguments were used 
to analyze pulse dynamics by Ei \cite{Ei02} for reaction-diffusion systems 
and Promislow \cite{Pr02} for damped Schr\"odinger equations.
See also \cite{Miz03} for a result for gKdV equations.
Although here we merely analyze stability of a single solitary wave,
we need these arguments because a general perturbation in the energy 
space may create a divergent phase shift of solitary waves.

To prove Lemma \ref{lem:lineardecay2}, we need to compare weighted norms of
$w$ in recentered moving coordinates. 
Recall that $\tau_h$ is a translation operator defined by
$(\tau_hf)(x):=f(x-h)$.
\begin{claim}
\label{cl:apb2}
Let $c_0>1$ and $\alpha\in(0,\alpha_{c_0})$.
There exists positive constants $\delta_6$, $\delta_7$ and $C_0$
such that if $|c-c_0|<\delta_6$ and $|h|<\delta_7$,
then for any $v\in H^1_\alpha$ with $P_cv=0$,
$$C_0^{-1} \|v\|_{H^1_\alpha}\le \|Q_{c_0}\tau_h v \|_{H^1_\alpha}
\le C_0\|v\|_{H^1_\alpha}.$$
\end{claim}
\begin{proof}
Since $P_cv=0$ we have
\begin{align*}
\|P_{c_0}\tau_hv\|_{H^1_\alpha}\le & \|P_{c_0}(\tau_hv-v)\|_{H^1_\alpha}
+ \|(P_{c_0}-P_c)v\|_{H^1_\alpha}
 \lesssim  (|c-c_0|+|h|)\|v\|_{L^2_\alpha} .
\end{align*}
Combining this with
$\|\tau_hv\|_{H^1_\alpha}=e^{\alpha h}\|v\|_{H^1_\alpha}$ and
$$
\left|\|Q_{c_0}\tau_hv\|_{H^1_\alpha}-\|\tau_hv\|_{H^1_\alpha}\right|
\le \|P_{c_0}\tau_hv\|_{H^1_\alpha}\,,$$
we have Claim \ref{cl:apb2}.
\end{proof}
\begin{proof}[Proof of Lemma \ref{lem:lineardecay2}]
To handle the time dependent advection term $\eta(t)\pd_yw$, we use
a sequence of coordinate frames moving with the constant speed $c_0$,
changing the phase from time to time so that the center of coordinates 
remains close to the solitary wave position $x(t)$ for all time.  

Let the constants $K_1$ and $\beta$ be as given by
Proposition~\ref{prop:Hnlinear}, and let $\delta_6$, $\delta_7$ and $C_0$
be from Claim~\ref{cl:apb2}. We fix $T_1>0$ such that 
\begin{equation}\label{i.T1}
  C_0^2K_1e^{-\beta T_1/6}\le 1,
\end{equation}
and let $t_j=jT_1$ for $j\ge0$.
Also let $h_j(t):=\int_{t_j}^t\bigl(c(s)-c_0+\eta(s)\bigr)\,ds$.
Under assumption \qref{a.ceta}, 
for $\hat\delta$ small enough we have that for every $j$,
\begin{equation}
  \label{eq:tj1}
\sup_{t\in[t_j,t_{j+1}]}|h_j(t)|\le T_1\hat\delta\le \delta_7 .
\end{equation}

Now let $ w_j(t)=Q_{c_0}\tau_{h_j(t)}w(t)$.
We rewrite \eqref{eq:linearw} as
\begin{equation*}
  \pd_tw_j=\mathcal{L}_{c_0}w_j+Q_{c_0}\tau_{h_j(t)}(F(t)+\widetilde{F}(t)),
\end{equation*}
where $\widetilde{F}(t)=
(f'(u_{c(t)})-\tau_{-h_j(t)}f'(u_{c_0})\tau_{h_j(t)})w(t)$.
Using the variation of constants formula, we have
\begin{equation}
  \label{eq:integral}
w_j(t)=e^{(t-t_j)\mathcal{L}_{c_0}}Q_{c_0}w_j(t_j)
+\int_{t_j}^te^{(t-s)\mathcal{L}_{c_0}}Q_{c_0}\tau_{h_j(s)}
(F(s)+\widetilde{F}(s))\,ds.
\end{equation}
By \eqref{eq:tj1} and the definition of $H^1_\alpha$,
the operator norm $\|\tau_{h_j(s)}\|_\ap \le 
e^{\alpha\delta_7}$.
Since $P_{c(t)}w(t)=0$, 
Claim \ref{cl:apb2} implies that for $t\in[t_j,t_{j+1}]$,
\begin{equation}
  \label{eq:equiv2}
C_0^{-1}\|w(t)\|_{H^1_\alpha}\le \|w_j(t)\|_{H^1_\alpha}\le 
C_0 \|w(t)\|_{H^1_\alpha},
\end{equation}
whence
\begin{align*}
\|\widetilde F(s)\|_{H^1_\alpha} \lesssim 
(|c(t)-c_0|+|h_j(t)|) \|w_j(t)\|_{H^1_\alpha}
\le (1+T_1)\hat\delta \|w_j(t)\|_{H^1_\alpha}.
\end{align*}
Applying Proposition~\ref{prop:Hnlinear} to 
\eqref{eq:integral} and using the estimates above, we have
\begin{align*}
& e^{\beta t}\|w_j(t)\|_{H^1_\alpha}\le 
K_1e^{\beta t_j}\|w_j(t_j)\|_{H^1_\alpha}
 +K_1 e^{\alpha\delta_7}
\int_{t_j}^te^{\beta s}
\|F(s)+\widetilde{F}(s)\|_{H^1_\alpha}\,ds
\\ &\quad  \le K_1e^{\beta t_j}\|w_j(t_j)\|_{H^1_\alpha}
+C_1\int_{t_j}^t e^{\beta s}
\left(\|F(s)\|_{H^1_\alpha}
+\hat\delta
\|w_j(s)\|_{H^1_\alpha}\right)\,ds
\end{align*}
for $t\in[t_j,t_{j+1}]$, where $C_1$ is a constant
independent of $j$.
Supposing $C_1\hat\delta\le\beta/2$, 
by Gronwall's inequality we infer
that for $t\in[t_j,t_{j+1}]$, 
\begin{equation}
  \label{eq:wjnorm}
e^{\beta t/2}\|w_j(t)\|_{H^1_\alpha}\le 
K_1e^{\beta t_j/2}\|w_j(t_j)\|_{H^1_\alpha}
+C_1\int_{t_j}^t e^{\beta s/2}\|F(s)\|_{H^1_\alpha}\,ds.
\end{equation}
By using \eqref{eq:equiv2} and then using
\qref{eq:wjnorm} with $j$ replaced by $j-1$, 
we find
\begin{align}
  \label{eq:induction1}
& e^{\beta t_j/2}
\|w_j(t_j)\|_{H^1_\alpha}
 \le 
e^{\beta t_j/2}
C_0^2\|w_{j-1}(t_j)\|_{H^1_\alpha}
\nonumber\\
&\quad \le 
C_0^2K_1
e^{\beta t_{j-1}/2}
\|w_{j-1}(t_{j-1})\|_{H^1_\alpha}
+ C_0^2C_1\int_{t_{j-1}}^{t_j} e^{\beta s/2}\|F(s)\|_{H^1_\alpha}\,ds.
\end{align}
Now $C_0^2K_1 
\le 
e^{\beta (t_j-t_{j-1})/6}$ due to \qref{i.T1}, 
and $ e^{\beta s/2} \le e^{\beta s/3} e^{\beta t_j/6} $
for $s\in[0,t_j]$, hence
\begin{align}
e^{\beta t_j/3}\|w_j(t_j)\|_{H^1_\alpha} &\le 
e^{\beta t_{j-1}/3}\|w_{j-1}(t_{j-1})\|_{H^1_\alpha}
+ C_0^2 C_1
\int_{t_{j-1}}^{t_j} e^{\beta s/3}\|F(s)\|_{H^1_\alpha}\,ds
\nonumber
\\ & \le \|w_0(0)\|_{H^1_\ap} + C_0^2C_1
\int_0^{t_j} e^{\beta s/3}\|F(s)\|_{H^1_\alpha}\,ds,
\label{eq:ind2}
\end{align}
by induction.
Combining \qref{eq:ind2} with \qref{eq:wjnorm} and \qref{eq:equiv2}
yields the conclusion of the Lemma.
\end{proof}